\documentclass[a4paper,UKenglish,cleveref, autoref, thm-restate]{lipics-v2021}

\pdfoutput=1 %uncomment to ensure pdflatex processing (mandatatory e.g. to submit to arXiv)
\hideLIPIcs  %uncomment to remove references to LIPIcs series (logo, DOI, ...), e.g. when preparing a pre-final version to be uploaded to arXiv or another public repository

\title{Optimising expectation with guarantees for window mean payoff in Markov decision processes}
\titlerunning{Optimising expectation with guarantees for window mean payoff} %TODO optional, please use if title is longer than one line

\author{Pranshu Gaba}{Tata Institute of Fundamental Research, Mumbai, India}{pranshu.gaba@tifr.res.in}{https://orcid.org/0009-0000-8012-780X}{}
\author{Shibashis Guha}{Tata Institute of Fundamental Research, Mumbai, India}{shibashis.guha@tifr.res.in}{https://orcid.org/0000-0002-9814-6651}{}

\authorrunning{P. Gaba and S. Guha}

 %TODO mandatory: Please choose ACM 2012 classifications from https://dl.acm.org/ccs/ccs_flat.cfm 

 %TODO mandatory; please add comma-separated list of keywords

\category{} %optional, e.g. invited paper

\relatedversion{Full version of paper accepted in AAMAS 2025.} %optional, e.g. full version hosted on arXiv, HAL, or other respository/website
\funding{This work is partially supported by the Indian Science and Engineering Research Board (SERB) grant SRG/2021/000466 and by the Indo-French Centre for the Promotion of Advanced Research (IFCPAR).} %optional, to capture a funding statement, which applies to all authors. Please enter author specific funding statements as fifth argument of the \author macro.

\nolinenumbers %uncomment to disable line numbering

\usepackage{mypackages}
\usepackage{mymacros}

\begin{document}

\maketitle

\begin{abstract}
    The window mean-payoff objective strengthens the classical mean-payoff objective by computing the mean-payoff over a finite window that slides along an infinite path. Two variants have been considered: in one variant, the maximum window length is fixed and given, while in the other, it is not fixed but is required to be bounded. In this paper, we look at the problem of synthesising strategies in Markov decision processes that maximise the window mean-payoff value in expectation, while also simultaneously guaranteeing that the value is above a certain threshold. We solve the synthesis problem for three different kinds of guarantees: sure (that needs to be satisfied in the worst-case, that is, for an adversarial environment), almost-sure (that needs to be satisfied with probability one), and probabilistic (that needs to be satisfied with at least some given probability $p$).
    
    We show that for fixed window mean-payoff objective, all the three problems are in $\mathsf{PTIME}$, while for bounded window mean-payoff objective, they are in $\mathsf{NP} \cap \mathsf{coNP}$, and thus have the same complexity as for maximising the expected performance without any guarantee. Moreover, we show that pure finite-memory strategies suffice for maximising the expectation with sure and almost-sure guarantees, whereas, for maximising expectation with a probabilistic guarantee, randomised strategies are necessary in general.
\end{abstract}

\section{Introduction}
\paragraph*{Beyond worst-case synthesis.}
Classical two-player quantitative zero-sum games~\cite{AG11,GTW02} involve decision making against a purely antagonistic environment, where a minimum performance needs to be guaranteed even in the worst case.
On the other hand, Markov decision processes (MDPs)~\cite{Puterman94} model uncertainty, and decision making involves ensuring a higher expected performance against a stochastic environment which usually does not provide any guarantee on the worst-case performance.
Both these models have their own weaknesses.
A strategy against an adversarial environment may provide worst-case guarantee but may be suboptimal in its expected behaviour.\
On the other hand, a strategy that maximises the expected performance may fail miserably in the worst-case situation.

However, in practice, both might be desired simultaneously: A system needs to provide guarantee in the worst-case, and perform well in an expected sense against a stochastic environment.
In~\cite{BFRR17}, the beyond worst-case (BWC) framework was introduced to provide strict worst-case guarantee as well as good expected performance.
In particular, the study was made for two quantitative objectives: mean payoff and shortest path.
While this work focussed on the restricted class of finite memory strategies, it was also shown that infinite memory strategies are strictly more powerful than finite-memory strategies in the BWC setting~\cite{BFRR17}.
The synthesis of infinite memory strategies was subsequently studied in~\cite{CR15}.

\paragraph*{Window mean payoff.}
For Boolean and quantitative prefix-independent objectives specified as the limit of a reward function in the long run~\cite{EM79, ZP96}, a play may satisfy such an objective and yet also exhibit undesired behaviours for arbitrarily long intervals~\cite{CHH09,CDRR15}. 
For instance, consider the following infinite sequence of payoffs: \((-1)\) \((+1)\) \((-1) (-1)\) \((+1)(+1)\) \((-1) (-1) (-1)\) \( (+1) (+1) (+1) \ldots\), where we see the  \((-1)\) payoff \(n\) times and then \((+1)\) payoff \(n\) times as \(n\) goes from \(1\) to \(\infty\). The classical (liminf) mean-payoff value of this sequence is zero, but there are arbitrarily long infixes in this sequence with mean-payoff \(-1\) (less than zero).
Finitary or window objectives strengthen such prefix-independent objectives by restricting the undesired behaviour to intervals  of bounded length (windows) of the play.
As a particular case, we consider the \emph{window mean-payoff objectives}, which are finitary versions of the classical mean-payoff objective. 
Window mean-payoff objectives~\cite{CDRR15} are quantitative finitary objectives
that strengthen the classical mean-payoff objective: the satisfaction of a window mean-payoff objective implies the satisfaction of the classical mean-payoff objective.
Given a length \(\WindowLength \geq 1\), the fixed window mean-payoff objective (\(\FWMP(\WindowLength, \Threshold)\)) is satisfied if except for a finite prefix, from every point in the play, there exists a window of length at most \(\WindowLength\) starting from that point such that the mean payoff of the window is at least a given threshold \(\Threshold\).
In the bounded window mean-payoff objective (\(\BWMP(\Threshold)\)), it is sufficient that there exists some length \(\WindowLength\) for which the \(\FWMP(\WindowLength, \Threshold)\) objective is satisfied.
The value of an outcome run is the largest~(supremum)~$\gamma$ such that from some point on, the run can be decomposed into windows of size smaller than the fixed or existentially quantified bound~$\WindowLength$, all having a mean-payoff value at least~$\gamma$.

\paragraph*{Contributions.}
In this paper, we study three different problems to maximise expectation while simultaneously providing guarantees for the fixed and the bounded window mean-payoff objectives.
Given an MDP and two threshold values \(\GuaranteeThreshold, \ExpectationThreshold \in \Rationals\), synthesise a strategy that
\begin{enumerate}
    \item (Beyond worst case (\(\BWC\)) synthesis) (i) ensures a window mean-payoff at least \(\GuaranteeThreshold\) surely, i.e. against all strategies of an adversarial environment, and (ii) an expectation that is at least \(\ExpectationThreshold\) against a stochastic model of the environment.
    \item (Beyond probability threshold (\(\BP\)) synthesis) (i) ensures a window mean-payoff larger than \(\GuaranteeThreshold\) with at least some given probability \(\ProbabilityGuaranteeThreshold\), and (ii) an expectation that is at least \(\ExpectationThreshold\) against a stochastic model of the environment.
    \item (Beyond almost sure (\(\BAS\)) synthesis) (i) ensures a window mean-payoff at  at least \(\GuaranteeThreshold\) almost surely, i.e. with probability one, and (ii) an expectation that is at least \(\ExpectationThreshold\) against a stochastic model of the environment.
\end{enumerate}
\emph{Motivating examples.}
We consider some motivating examples in the context of window mean payoff for maximising expectation while providing guarantees at the same time.
\begin{itemize}
    \item \emph{Consistent output of power plant.} 
    A power plant may be required to output 10 MW of power on average. 
    A plant that outputs 240 MW of power for an hour and 0 MW for the rest of the day satisfies the requirement but is not desirable if there is no means to store the energy. 
    Using windows, we can ask for an output of 10 MW every hour.
    Moreover, a power plant may have multiple ways to generate power (solar, wind, hydro) with different rates and reliability, e.g. solar output is high during the day and low during the night, while hydro is consistently low-moderate. 
    It is may be desired to devise a strategy that maximises the output in expectation while producing sufficient power at all times for critical applications such as hospitals and trains. 
    \item \emph{Investment in stock market.}
    While investing in the stock market, an investor not only wants a higher expected return, but may as well prefer to be risk-averse, that is, the stocks do not crash or such events happen rarely.
    Further, the investor may consistently want to receive returns over a certain period or a window of time.
    \item \emph{Gambling}~\cite{CENR18}. 
    In gambling too, while the goal is to maximise the expected profit, a desirable policy may want to avoid risk, and thus would try to ensure that the chances of losing is less than a certain probability.
    Further, a gambler would like to receive a profit amount from time to time and the payment should not be deferred indefinitely.
\end{itemize}
Thus the problems considered in this paper are important ones.

We show that, for the case of fixed window mean-payoff, all of the above problems are in \(\PTime\) (Theorems~\ref{thm:fwmp-sure-result}, \ref{thm:fwmp-probability-result}, and \ref{thm:fwmp-almost-sure-result}) and are thus no more complex than solving two-player games~\cite{CDRR15} or maximising expectation for the same objective in an MDP~\cite{BGR19}.
For classical mean-payoff objective, we note that the first problem above is in \(\NP \cap \coNP\)~\cite{BRR17} while the second and the third problems are in \(\PTime\)~\cite{CR15}.
For the case of bounded window mean-payoff objective, all the above problems are in \(\NP \cap \coNP\) (\Cref{thm:BWMP}), thus showing that the results are no more complex than solving two-player games or maximising expectation for the same objective in an MDP~\cite{CDRR15,BGR19}.

Our techniques are different from the BWC synthesis for classical mean-payoff objective.
An \emph{end-component} (EC) \(\EC\) is a strongly connected component of an MDP such that all outgoing edges of each probabilistic vertex in \(\EC\) also lead to a vertex in \(\EC\).
A maximal end-component (MEC) is an EC which is not included in any other end-component.
A winning end-component (WEC) is an EC such that from every vertex there exists a winning strategy for the sure objective while staying inside the end-component.
A maximal winning end-component is a winning end-component that is not included in any other winning end-component.
For sure-expectation problem for classical mean payoff~\cite{BFRR17}, maximal winning end-components (MWEC) are constructed. 
Within each MWEC, there are two strategies, one for the sure mean-payoff objective and the other for the expected mean-payoff objective and it may be required that the strategies alternate to meet both the sure and the expected thresholds.
In the current work, for window mean-payoff objective, for the sure-expected problem, we do not need to find the MWECs but instead it suffices to compute the maximum sure window mean payoff value from each vertex.
Further, unlike mean-payoff objective~\cite{BFRR17}, for window mean-payoff objective infinite memory strategies are no more powerful than finite memory strategies and we need to switch from a strategy corresponding to the expectation maximisation to the sure satisfaction strategy only once.

\paragraph*{Related Work.}
The beyond worst-case framework was introduced in~\cite{BFRR17} for quantitative objectives.
The problem was studied for finite-memory strategies and it was shown to be in \(\NP \cap \coNP\) for mean-payoff objective.
The case of infinite memory strategy for the \(\BWC\) synthesis problem was left open in~\cite{BFRR17} and was solved in~\cite{CR15}.
Further, in~\cite{CR15}, a natural relaxation of the \(\BWC\) problem, the beyond almost-sure synthesis problem (\(\BAS\)) was introduced and was shown to be in \(\PTime\) for mean-payoff objective.
The beyond probability threshold synthesis problem was studied for mean-payoff objective in~\cite{CKK17}.
In~\cite{BRR17}, the \(\BWC\), \(\BAS\), and \(\BP\) synthesis problems were studied for qualitative omega-regular objectives encoded as parity objectives.
The problems were shown to be in \(\NP \cap \coNP\).
In~\cite{CP19}, the above problem was studied in the context of stochastic games which are a generalisation of MDPs where the environment is both stochastic and adversarial.
A combination of optimising expected mean payoff and surely satisfying omega-regular~\cite{AKV16}, safety~\cite{GGR18}, and energy objectives~\cite{BKN16} were also considered.
In~\cite{BGR20}, Boolean combinations of objectives that are omega-regular properties that need to be enforced either surely, almost surely, existentially, or with non-zero probability were studied.
It was shown that both randomisation and infinite memory may be required by an optimal strategy.
In~\cite{BKW24}, a combination of parity objective and multiple reachability objectives along with threshold probabilities were considered where the parity objective needs to be satisfied surely and each reachability objective is satisfied with the corresponding threshold probability.
The \(\BWC\) and the \(\BP\) problems were also studied for the discounted-sum objective in partially observed MDPs (POMDPs)~\cite{CNPRZ17, CENR18}.

Some of the above works used different kinds of end-components in their reasoning.
While~\cite{BFRR17} used MWECs, in~\cite{CR15} and \cite{CKK17} MECs are considered, and in~\cite{AKV16} and ~\cite{BRR17} sophisticated kind of end-components called super-good end-components and ultra good end-components were used.

Mean-payoff objectives were studied initially in two-player games, without stochasticity~\cite{EM79,ZP96}, and finitary versions were introduced as window mean-payoff objectives~\cite{CDRR15}.
For finitary mean-payoff objectives, the satisfaction problem~\cite{BDOR20} and the expectation problem~\cite{BGR19} were studied in Markov decision processes (MDPs).
The technique in both papers relies on first identifying end components of the given MDP.
While in~\cite{BDOR20}, the probability to reach the `good' end-components is maximised, in~\cite{BGR19}, the end components are collapsed, and in the modified MDP, expected mean payoff is maximised.
Both the expectation problem~\cite{BGR19} and the satisfaction problem~\cite{BDOR20} for the \(\FWMPL\) objective are in $\PTime$, while they are in \(\UP \cap \coUP\) for the \(\BWMP\) objective.
The satisfaction problem for window mean-payoff objectives has been studied recently in~\cite{DGG24} for stochastic games.
There, while satisfaction with positive probability and almost-sure satisfaction, that is, with probability \(1\), of \(\FWMPL\) are in $\PTime$, the problem is in \(\UP \cap \coUP\) for quantitative satisfaction\footnote{In~\cite{DGG24}, quantitative satisfaction is shown to be in \(\NP \intersection \coNP\). 
The probability of satisfaction from each vertex is unique and is polynomial in the size of the input, giving \(\UP \intersection \coUP\) membership.
}, i.e., with threshold probabilities $p \not\in \{0,1\}$.
Furthermore, the satisfaction problem of \(\BWMP\) is in \(\UP \cap \coUP\) and thus has the same complexity as that of the special case of MDPs.

In the current work too, we analyse MECs but in a way that is different from the above works.
While pure finite-memory strategies suffice for the sure-expectation and the almost sure-expectation problems for the window mean-payoff objectives, for the satisfaction-expectation problem, we need finite-memory randomised strategies.

\paragraph*{Organisation.}
In \Cref{sec:preliminaries}, we define the technical preliminaries.
In \Cref{sec:bwc-bas-bpt-definitions}, we formally introduce the three problems related to expectation maximisation with guarantees.
In \Cref{sec:fwmp} and in \Cref{sec:bwmp}, we study the problems for the fixed and bounded window mean-payoff objectives respectively.
We conclude in \Cref{sec:conc}.

\section{Preliminaries}%
\label{sec:preliminaries}
In this section, we define some necessary prerequisites.

\paragraph*{Probability distributions.}
For a finite set \(A\), a \emph{probability distribution} over \(A\) is a function \(\Prob \colon A \to [0,1]\) such that \(\sum_{a \in A} \Prob(a) = 1\).
We denote by \(\DistributionSet{A}\) the set of all probability distributions over \(A\).
The support of the probability distribution \(\Prob\) on \(A\) is \(\Support{\Prob} = \{a \in A \mid \Prob(a) > 0\}\).
The distribution \(\Prob\) is \emph{Dirac} if \(|\Support{\Prob}| = 1\).
For each \(a \in A\), we define a Dirac distribution \(\mathbf{1}_{a}(a')\) which equals \(1\) if \(a'=a\) and \(0\) otherwise.
For algorithm and complexity reasons, we assume that the probability distributions take rational values.

\paragraph*{Markov decision processes (MDPs).}
A \emph{Markov decision process} (MDP) is a tuple \(\MDP = ((\Vertices, \Edges), (\VerticesPlayer, \VerticesRandom), \ProbabilityFunction, \PayoffFunction)\) where:
\begin{itemize}
    \item \((V,E)\) is called the \emph{arena} of \(\MDP\). It is a directed graph with a set \(V\) of vertices and a set \(E \subseteq (\VerticesPlayer \times \VerticesRandom) \cup (\VerticesRandom \times \VerticesPlayer)\) 
    of edges such that for each vertex \(v \in V\), there is an out-edge from \(v\) in the game (i.e., no deadlocks).
    We denote by \(\OutNeighbours{\Vertex}\) the set of vertices \(u\) such that \((\Vertex, u) \in E\).
    We say that the MDP \(\MDP\) is finite if the set \(V\) is finite. 
    Unless mentioned otherwise, we consider MDPs to be finite in this work.
    \item \((\VerticesPlayer, \VerticesRandom)\)  is a partition of the set \(\Vertices\) of vertices, where \(\VerticesPlayer\) denotes the set of vertices belonging to the player and \(\VerticesRandom\) denotes probabilistic vertices. 
    \item \(\ProbabilityFunction \colon \VerticesRandom \to \DistributionSet{\VerticesPlayer}\) is the \emph{probability function} that returns the probability distribution over the out-neighbours of probabilistic vertices.
    We require for every probabilistic vertex \(v \in \VerticesRandom\) that \(\Support{\ProbabilityFunction(v)} = \Edges(v)\), that is, for all vertices \(v' \in \Vertices\), we have that \(\ProbabilityFunction(v)(v') > 0\) if and only if \(v'\) is an out-neighbour of \(v\).
    \item \(\PayoffFunction \colon E \to \Integers\) is the \emph{payoff function} that defines an integer payoff for every edge in the arena.
    Let \(W_\MDP\) be the maximum weight appearing on the edges in \(\MDP\).
    We drop the subscript when it is clear from the context.
\end{itemize}
With a little abuse of nomenclature, we mean by \emph{self-loop} from a vertex \(\Vertex\) a sequence of two edges starting from and ending at \(v\) so that player vertices and probabilistic vertices alternate.
A payoff \(\Threshold \) on the self-loop here denotes that both the edges that are part of the self-loop have the same payoff \(\Threshold\).

A \emph{run} of the MDP begins by placing a token on an initial vertex which is a player vertex and proceeds in steps.
In each step, if the token is on a player vertex \(\Vertex\), then the player chooses an out-edge of \(\Vertex\) and moves the token along that edge.
Otherwise, if the token is on a probabilistic vertex \(\Vertex\), then the out-edge is chosen by the probability distribution \(\Prob(v)\).
This continues \emph{ad infinitum}, resulting in a run \(\Run\) that is an infinite path in the arena. 

For a run \(\Run = v_{0} v_{1} v_{2} \cdots\), we denote by \(\Run(i)\) the vertex \(v_i\), by \(\RunInfix{i}{j}\) the infix \(v_{i} \cdots v_{j}\), by \(\RunPrefix{j}\) the finite prefix \(v_{0} v_{1} \cdots v_{j}\), and by \(\RunSuffix{i}\) the infinite suffix \(v_{i} v_{i+1} \cdots\).
The length of an infix \(\RunInfix{i}{j}\) is the number of edges, that is \(j - i\), and is denoted by \(\abs{\RunInfix{i}{j}}\).
We denote by \(\RunSet^\MDP\), \(\PrefixSet^\MDP\), and \(\PrefixSet^\MDP_\Player\) the set of all runs, the set of all finite prefixes in \(\MDP\), and the set of all finite prefixes in \(\MDP\) ending in a vertex in \(\VerticesPlayer\) respectively.
We drop the superscript \(\MDP\) when they are clear from the context.
We denote by \(\inf(\Run)\) the set of vertices in \(V\) that occur infinitely often in \(\Run\).

An MDP where every vertex in \(\VerticesPlayer\) has exactly one out-neighbour is called a \emph{Markov chain} and an MDP where every vertex in \(\VerticesRandom\) has exactly one out-neighbour is called a \emph{one-player game}.

\Cref{fig:bwc-example} shows an example of an MDP.
In figures, in MDPs, we denote player vertices by circles and probabilistic vertices by diamonds.

\begin{figure}[t]
    \centering
    \begin{tikzpicture}
        \node[state] (v0) {\(v_{0}\)};
        \node[random, draw, right of=v0] (v1) {\(v_{1}\)};
        \node[state, right of=v1] (v2) {\(v_{2}\)};
        \node[random, draw, right of=v2] (v3) {\(v_{3}\)};
        \node[state, right of=v3] (v4) {\(v_{4}\)};
        \node[random, draw, above right of=v4] (v5) {\(v_{5}\)};
        \node[random, draw, below right of=v4] (v6) {\(v_{6}\)};
        \node[state, above right of=v6] (v7) {\(v_{7}\)};
        \node[random, draw, right of=v7] (v8) {\(v_{8}\)};
        \draw 
              (v0) edge[bend left] node[above, pos=0.3]{\(\EdgeValues{-1}{}\)} (v1)
              (v1) edge[bend left] node[below, pos=0.3]{\(\EdgeValues{+2}{.3}\)} (v0)
              
              (v1) edge[bend left] node[above, pos=0.3]{\(\EdgeValues{+1}{.7}\)} (v2)
              (v2) edge[bend left] node[below, pos=0.3]{\(\EdgeValues{0}{}\)} (v1)
              
              (v2) edge[bend left] node[above, pos=0.3]{\(\EdgeValues{0}{}\)} (v3)
              (v3) edge[bend left] node[below, pos=0.3]{\(\EdgeValues{-1}{.8}\)} (v2)
              
              (v3) edge node[above, pos=0.3]{\(\EdgeValues{+4}{.2}\)} (v4)
             
              (v4) edge node[above left, pos=0.3]{\(\EdgeValues{0}{}\)} (v5)
              (v5) edge node[above right, pos=0.3]{\(\EdgeValues{0}{1}\)} (v7)
              (v7) edge node[below right, pos=0.3]{\(\EdgeValues{0}{}\)} (v6)
              (v6) edge node[below left, pos=0.3]{\(\EdgeValues{+10}{1}\)} (v4)
              
              (v7) edge[bend left] node[above, pos=0.3]{\(\EdgeValues{+6}{}\)} (v8)
              (v8) edge[bend left] node[below, pos=0.3]{\(\EdgeValues{-4}{1}\)} (v7)
        ;
    \end{tikzpicture}
    \caption{An example of an MDP.}
    \label{fig:bwc-example}
\end{figure}
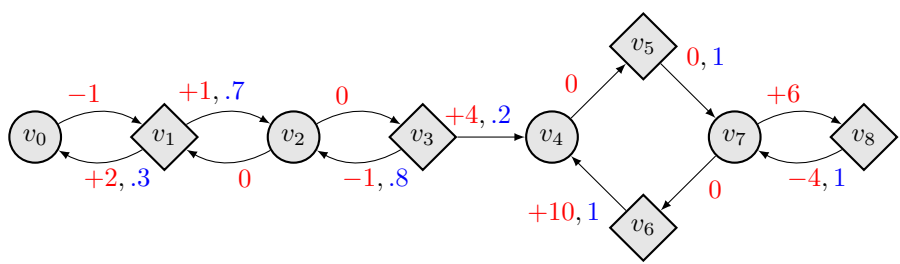

\paragraph*{Boolean objectives.}
Depending on the specifications, some runs are desirable for the player, and some are not.
Objectives can be broadly classified into two types:
Boolean or quantitative. 
A \emph{Boolean objective} \(\Objective\) is a set of runs that are desirable for the player.
We say a run \(\Run \in \RunSet\) \emph{satisfies} an objective \(\Objective\) if \(\Run \in \Objective\).
Given a set \(T \subseteq V\) of target vertices, a common Boolean objective is the \emph{reachability objective}, defined as \(\Reach{T} = \{\Run \in \RunSet \:|\: \exists i \ge 0, \Run(i) \in T\}\), i.e., the set of runs that visit \(T\).
Some other examples of Boolean objectives include safety, \Buchi, co\Buchi, and parity.

\paragraph*{Quantitative objectives.}
A \emph{quantitative objective} is a function \(\Objective \colon \RunSet \to \Rationals \cup \{\pm \infty\}\) that assigns to each run in the MDP a numerical value that denotes how good the run is for the player.
Some common examples of quantitative objectives include mean-payoff, discounted-sum payoff, energy payoff, total payoff and liminf payoff.
For a run \(\Run = v_0 v_1 v_2 \cdots\), the liminf mean-payoff objective is defined as follows:
\(\ObjectiveMP(\Run) \define \liminf_{n \to \infty} \frac{1}{n} \sum_{i=0}^{n} \PayoffFunction(v_i, v_{i+1})\).
The objectives studied in this paper, the window mean-payoff objectives (defined later in this section), are also quantitative objectives.
Corresponding to a quantitative objective \(\Objective\), we define \emph{threshold Boolean objectives} \(\{\Run \in \RunSet \mid \Objective(\Run) \ge \Threshold \}\), for thresholds \(\Threshold \in \Reals\).
We denote these objectives succinctly as \(\ThresholdObjective{\ge \Threshold}\).

A quantitative objective \(\Objective\) is \emph{closed under suffixes} if for all runs \(\Run\), and for all suffixes \(\RunSuffix{j}\) of \(\Run\), we have that \(\Objective(\Run) = \Objective(\RunSuffix{j})\). 
An objective \(\Objective\) is \emph{closed under prefixes} if for all runs \(\Run \) and all prefixes \(\Prefix\) such that \(\Prefix \cdot \Run \in \RunSet^{\MDP}\), we have that \(\Objective(\Run) = \Objective(\Prefix \cdot \Run)\).
An objective \(\Objective\) is \emph{prefix-independent} if it is closed under both prefixes and suffixes.
Mean-payoff is an example of prefix-independent objectives.
Prefix independence can be defined analogously for Boolean objectives.

\paragraph*{Strategies.}
A \emph{strategy} for the player in an MDP \(\MDP\) is a function \(\Strategy: \PrefixSet_\Player \to \DistributionSet{\Vertices}\) that maps prefixes ending in a vertex \(v \in \VerticesPlayer\) to a distribution over the successors of \(v\). 
The set of all strategies of the player in the MDP \(\MDP\) is denoted by \(\StrategySet{\MDP}\).
Strategies can be realised as the output of a (possibly infinite-state) Mealy machine.
A \emph{Mealy machine} is a deterministic transition system with transitions labelled by input/output pairs. 
Formally, a Mealy machine \(M\) is a tuple \((Q, q_0, \Sigma_i, \Sigma_o, \Delta, \delta)\) 
where 
\(Q\) is the set of states of \(M\) (the memory of the induced strategy), \(q_0 \in Q\) is the initial state,
\(\Sigma_i\) is the input alphabet, 
\(\Sigma_o\) is the output alphabet,
\(\Delta \colon Q \times \Sigma_i \to Q\) is a transition function that reads the current state of~\(M\) and an input letter and returns the next state of  \(M\), and \(\delta \colon Q \times \Sigma_i \to \Sigma_o \) is an output function that reads the current state of
\(M\) and an input letter and returns an output letter.
The transition function \(\Delta\) can be extended to a function \(\widehat{\Delta} \colon Q \times \Sigma_i^+ \to Q\) that reads words and can be defined inductively by \(\widehat{\Delta}(q, a) = \Delta(q, a)\) and \(\widehat{\Delta}(q, x \cdot a) = \Delta(\widehat{\Delta}(q, x), a) \), for \(q \in Q\), \(x \in \Sigma_i^+\), and \(a \in \Sigma_i\).

A strategy of the player can be defined by a Mealy machine \(M = (Q, q_0, \Vertices, \DistributionSet{\Vertices}, \Delta, \delta)\) as follows: 
Given a prefix \(\Prefix' \cdot \Vertex \in \PrefixSet_\Player\) ending in a player vertex \(\Vertex\), the strategy \(\Strategy\) defined by the Mealy machine \(M\) is \(\Strategy(\Prefix' \cdot \Vertex) = \delta(\widehat{\Delta}(q_0, \Prefix'), \Vertex)\).
Intuitively, in each step, if the token is on a vertex~\(v\) that belongs to the player, then~\(v\) is given as input to the Mealy machine, and the Mealy machine outputs a distribution over the successor probabilistic vertices of~\(v\) that the player must choose. 
Otherwise, the token is on a vertex \(v\) that is a probabilistic vertex, in which case, the Mealy machine outputs the distribution \(\ProbabilityFunction(v)\) that is part of the MDP \(\MDP\).

A strategy is \emph{deterministic} if for every prefix \(\Prefix \in \PrefixSet_\Player\), we have that \(\Strategy(\Prefix)\) is Dirac, otherwise, it is \emph{randomised}. 
The \emph{memory size} of a strategy \(\Strategy\) is the smallest number of states a Mealy machine defining \(\Strategy\) can have. 
A strategy \(\Strategy\) is \emph{memoryless} if \(\Strategy(\Prefix)\) only depends on the last element of the prefix~\(\Prefix\), that is, for all prefixes \(\Prefix, \Prefix' \in \PrefixSet_\Player\) if \(\Last{\Prefix} = \Last{\Prefix'}\), then \(\Strategy(\Prefix) = \Strategy(\Prefix')\).
Memoryless strategies can be defined by Mealy machines with only one state. 
Fixing a strategy \(\Strategy\) of the player in an MDP yields a (possibly infinite-state) Markov chain, and we represent this by \(\MDP^{\Strategy}\).
Finally, if the same strategy can be used regardless of the initial state, we call it a \emph{uniform strategy}.

A run \(\Run = v_0 v_1 \dotsm\) is \emph{consistent} with a strategy \(\Strategy \in \StrategySet{\MDP}\) if for all \(j \geq 0\) with \(v_{j} \in \Vertices\), we have \(v_{j+1} \in \Support{\Strategy(\RunPrefix{j})}\). 
A run \(\Run\) is an \emph{outcome} of a strategy \(\Strategy\) if \(\Run\) is consistent with \(\Strategy\). 
We denote by \(\Outcomev{\Vertex}{\MDP}{\Strategy}\) the set of runs of \(\MDP\) that start from \(\Vertex\) and are consistent with strategy \(\Strategy\).

\paragraph*{Satisfaction probability of Boolean objectives.}
For a Boolean objective \(\Objective\), we denote by \(\Pr{\Strategy}{\MDP, v}{\Objective}\) the probability that an outcome of the strategy \(\Strategy\) in \(\MDP\) with initial vertex \(v\) satisfies \(\Objective\). 
The \emph{cone} at \(\Prefix\) is the set \(\Cone{\Prefix} \define \{\Run \in \RunSet^{\MDP} \suchthat \Prefix \text{ is a prefix of } \Run \}\), the set of all runs having \(\Prefix\) as a prefix. 
First, we define this probability measure over cones inductively as follows. 
If \(\abs{\Prefix} = 0\), then \(\Prefix \) is just a vertex \(v_{0}\), and 
\(\Pr{\Strategy}{\MDP, v}{\Cone{\Prefix}}\) is \(1\) if \(v = v_0\), and \(0\) otherwise. 
For the inductive case \(\abs{\Prefix} > 0\), 
there exist \(\Prefix' \in \PrefixSet\) and \(v' \in \Vertices\) such that \(\Prefix = \Prefix' \cdot v'\), and we have \(\Pr{\Strategy}{\MDP, v}{\Cone{\Prefix' \cdot v'}} = \Pr{\Strategy}{\MDP, v}{\Cone{\Prefix'}}  \cdot \Strategy(\Prefix')(v')\).
As shown in~\cite{Vardi85}, it is sufficient to define \(\Pr{\Strategy}{\MDP, v}{\Objective}\) on cones in \(\MDP\) since a measure defined on cones extends to a unique measure on \(\RunSet\) by Carath\'{e}odory's extension theorem~\cite{Billingsley86}.

Given an MDP \(\MDP\) with a Boolean objective \(\Objective\),  starting from a vertex \(v\) in \(\MDP\), we are interested in finding the maximum probability with which the player can ensure that the objective \(\Objective\) is satisfied.
In the decision problem, we ask if given \(p \in [0, 1] \intersection \Rationals\), can the player ensure, with probability at least \(p\), that an outcome satisfies \(\Objective\).
Formally, we ask if there exists a strategy \(\Strategy \in \StrategySet{\MDP}\) of the player such that \(\Pr{\Strategy}{\MDP, v}{\Objective} \ge p\).

\paragraph*{Expected value of quantitative objectives.}
For a quantitative objective \(\Objective\), we are interested in determining the maximum value of \(\Objective\) that the player can ensure in expectation.
Formally, given a strategy \(\Strategy\) and an initial vertex \(v\), we denote by \(\ExpectationOutcome{\Strategy}{\MDP, v}{\Objective}\) the \emph{expected \(\Objective\)-value of an outcome} of \(\Strategy\) from \(v\), that is, the expectation of \(\Objective\) over all plays with initial vertex \(v\) under the probability measure \(\Pr{\Strategy}{\MDP, v}{\Objective}\). 
The expected \(\Objective\)-value of a vertex \(v\) is \(\ExpectationValue{\MDP}{v}{\Objective} =  \sup_{\Strategy} \ExpectationOutcome{\Strategy}{\MDP, v}{\Objective} \) is the supremum over all strategies \(\Strategy\) of the player of the expected \(\Objective\)-value of the outcomes from \(v\).
For all \(\epsilon > 0\), a strategy \(\Strategy\) is \(\epsilon\)-\emph{optimal} for objective \(\Objective\) if it is \(\epsilon\)-close to the value of the vertex, that is, if \(\ExpectationOutcome{\Strategy}{\MDP, v}{\Objective} \ge \ExpectationValue{\MDP}{v}{\Objective} - \epsilon\).
A strategy \(\Strategy\) is \emph{optimal} for objective \(\Objective\) if it achieves the value of the vertex, that is, if \(\ExpectationOutcome{\Strategy}{\MDP, v}{\Objective} = \ExpectationValue{\MDP}{v}{\Objective}\).

\paragraph*{Two-player games.}
Observe that an MDP can be seen as a game where \(\Player\) plays against a stochastic adversary and are sometimes called \(1\frac{1}{2}\)-player games.
We would also need to consider real two-player games in this paper.
An MDP \(\MDP = ((\Vertices, \Edges), (\VerticesPlayer, \VerticesRandom), \ProbabilityFunction, \PayoffFunction)\) can be seen as a two-player game, denoted \(\Game_\MDP = ((\Vertices, \Edges), (\VerticesPlayer, \VerticesRandom), \PayoffFunction)\) where the probability vertices in \(\VerticesRandom\) are interpreted as vertices belonging to an adversarial environment, the probability function \(\ProbabilityFunction\) is forgotten, and the adversary chooses a strategy of its choice.
Using both interpretations, which are that of a stochastic model as well as the adversarial environment is crucial to our work.

\paragraph*{Maximal end components.}
An \emph{end component} (EC) in an MDP is a subset \(T \subseteq V\) of vertices such that 
for every probabilistic vertex \(v\) in \(T\), every out-neighbour of \(v\) belongs to the subset \(T\), and
\(T\) is strongly connected, that is, for every pair of vertices \(v, v'\) in the subset \(T\), the player has a strategy to reach \(v'\) from \(v\) with probability~\(1\).
Intuitively, an EC \(T\) is a set of vertices such that
if the token is in \(T\), then the player can ensure with probability~\(1\) that the token never leaves \(T\), and that the player can almost surely visit every vertex in \(T\) from every other vertex in \(T\) without leaving \(T\).
A \emph{maximal end component} (MEC) is an EC that is not contained in any other EC.
The MECs in an MDP are disjoint.
Each vertex in an MDP belongs to either no MEC or exactly one MEC.
Thus, the number of MECs is bounded above by the number of vertices in the MDP.
We denote by \(\mathfrak{M}_\MDP\) the set of MECs of \(\MDP\).
The subscript \(\MDP\) is dropped when it is clear from the context.
The MEC decomposition can computed in polynomial time~\cite{CH14}.

We now recall some of the classical results on Markov Decision Processes that will be used later.
\begin{lemma}[Optimal reachability~\cite{BK08}]%
\label{lem:optimalreach}
    Given an MDP \(\MDP\) and a set \(T \subseteq V\) of target states, we can compute  in polynomial time for each vertex \(v \in V\), the probability \(p^*_v = \sup_\Strategy \Pr{\Strategy}{\MDP, v}{\Reach{T}}\) with which the player can ensure visiting \(T\).
    There is an optimal uniform pure memoryless strategy \(\Strategy^*\) that enforces reaching \(T\) with probability \(p^*_v\) from every vertex \(v \in V\).
    
    Further, for all \(v \in V, c < p^*_v\), there exists \(N \in \NonNegativeNaturals\) such that by playing \(\Strategy^*\) for \(N\) steps, we reach \(T\) from \(\Vertex\) with probability greater than \(c\).
\end{lemma}

For an arbitrary strategy \(\Strategy\) of the player, it is almost-surely the case that an outcome ends up in an MEC of \(\MDP\).
\begin{lemma}[Long-run appearance in MECs~\cite{BK08}]%
\label{lem:longrunMEC}
    Given an MDP \(\MDP\) with a set \(\Vertices\) of vertices, for every strategy \(\Strategy\) of the player and for every vertex \(v \in \Vertices\), we have that \(\sum_{\MEC \in \MECs} \Pr{\Strategy}{\MDP, v}{\inf(\Run) \subseteq M}\) = 1.
\end{lemma}

\paragraph*{Window mean-payoff objectives.}
In this work, we look at the \(\BWC\) framework in the context of the window mean-payoff objectives.
We first define Boolean versions of the objective.

For a run \(\Run = v_{0} v_{1} v_{2} \cdots\) in an MDP \(\MDP\), the \emph{mean payoff} of an infix \(\RunInfix{i}{i+n} \)  is the average of the payoffs of the edges in the infix and  is defined as \(\mathsf{MP}(\RunInfix{i}{i+n}) = \sum_{k=i}^{i+n-1} \frac{1}{n} \PayoffFunction(\Vertex[k], \Vertex[k+1])\).
Given a window length \(\WindowLength \geq 1\) and a threshold \(\Threshold \in \Rationals\), a run \(\Run = \Vertex[0] \Vertex[1] \dotsm \)  in \(\MDP\) satisfies the \emph{fixed window mean-payoff objective} \(\FWMP_{\MDP}(\WindowLength, \Threshold)\) if from every position after some point, it is possible to start an infix of length at most $\WindowLength$ with mean payoff at least \(\Threshold\).
\begin{equation*}
    \FWMP_{\MDP}(\WindowLength, \Threshold) = \{ \Run \in \RunSet^\MDP \mid \exists k \geq 0 \cdot \forall i \ge k \cdot \exists j \in \PositiveSet{\WindowLength}: \mathsf{MP}(\Run(i, i + j)) \ge \Threshold\}
\end{equation*} 
We omit the subscript \(\MDP\) when it is clear from the context. 
Given a threshold \(\Threshold\), a run \(\Run = \Vertex[0] \Vertex[1] \cdots\), and \(0 \le i<  j\), we say that the \emph{\(\Threshold\)-window} \(\Run(i, j)\) is \emph{open} if the mean payoff of \(\Run(i,k)\) is less than \(\Threshold\) for all \(i < k \le j\). 
Otherwise, the \(\Threshold\)-window is \emph{closed}.
Given \(j > 0\), we say a \(\Threshold\)-window is open at \(j\) if there exists an open \(\Threshold\)-window \(\Run(i, j)\) for some \(i < j\). 
The \(\Threshold\)-window starting at position \(i\) \emph{closes} at position \(j\) if \(j\) is the first position after \(i\) such that the mean payoff of \(\Run(i, j)\) is at least \(\Threshold\). 
Corresponding to the Boolean objective \(\FWMP_{\MDP}(\WindowLength, \Threshold)\), we define a quantitative version of the objective as follows:
Given a run \(\Run\) in  an MDP \(\MDP\), the \(\ObjectiveFWMPL\)-value of \(\Run\) is equal to  \(\sup \{\Threshold \in \Reals \suchthat \Run \in \FWMP(\WindowLength, \Threshold) \} \), the supremum threshold~\(\Threshold\) such that the run satisfies \(\FWMP_{\MDP}(\WindowLength, \Threshold)\).
For a run \(\Run\) in an MDP, the \(\ObjectiveFWMPL\)-value of \(\Run\) can be of the form \(\frac{a}{b}\) where \(a \in \{-W \cdot \WindowLength, \dots, 0, \dots, W \cdot \WindowLength\}\) and \(b \in \{1, \dots, \WindowLength\}\).
Hence the \(\ObjectiveFWMPL\)-value can be one of finitely many values leading to the following.
\begin{proposition}%
\label{prop:FWMP}
    For all \(\WindowLength \ge 1\) and all \(\Threshold \in \Reals\), we have \(\Run \in \{\ObjectiveFWMPL \ge \Threshold\}\) if and only if \(\Run \in \FWMP(\WindowLength, \Threshold)\).
\end{proposition}

We also consider another window mean-payoff objective called the \emph{bounded window mean-payoff objective} \(\BWMP_\MDP(\Threshold)\).
A run satisfies the objective \(\BWMP_\MDP(\Threshold)\) if there exists a window length \(\WindowLength \ge 1\) such that the run satisfies \(\FWMP_\MDP(\WindowLength, \Threshold)\).
\[
    \BWMP_{\MDP}(\Threshold) = \{ \Run \in \RunSet^\MDP \suchthat \exists \WindowLength \ge 1 : \Run \in \FWMP_{\MDP}(\WindowLength, \Threshold)\}
\]
Equivalently, a run \(\Run\) does not satisfy \(\BWMP_\MDP(\Threshold)\) if and only if for every suffix of \(\Run\), for all \(\WindowLength \ge 1\), the suffix contains an open \(\Threshold\)-window of length \(\WindowLength\).

We define the \(\ObjectiveBWMP\)-value of a run \(\Run\) analogously to \(\ObjectiveFWMPL\).
Given a run \(\Run\) in  an MDP \(\MDP\), the \(\ObjectiveBWMP\)-value of \(\Run\) is equal to  \(\sup \{\Threshold \in \Reals \suchthat \Run \in \BWMP(\Threshold) \} \), or equivalently, \(\sup \{\Threshold \in \Reals \suchthat \exists \WindowLength \ge 1 : \Run \in \FWMP(\WindowLength, \Threshold) \} \).

An observation similar to Proposition~\ref{prop:FWMP} does not hold for the bounded window mean-payoff objective.
This is because since \(\WindowLength\) can be unbounded, there may be a run \(\Run\) such that \(\Run\) does not satisfy  \(\FWMP(\WindowLength, \ObjectiveBWMP(\Run) )\) for any \(\WindowLength \ge 1\).
However the following holds.
\begin{proposition}%
\label{prop:BWMP}
    For all \(\Threshold \in \Reals\), if \(\Run \in \BWMP(\Threshold)\), then \(\Run \in \{\ObjectiveBWMP \ge \Threshold\}\).
\end{proposition}
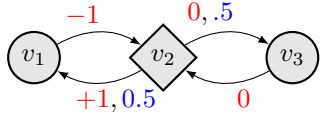
\begin{figure}[t]
    \centering
    \begin{tikzpicture}
        \node[state] (v1) {\(\Vertex[1]\)};
        \node[random, draw, right of=v1] (v2) {\(\Vertex[2]\)};
        \node[state, right of=v2] (v3) {\(\Vertex[3]\)};
        \draw 
              (v1) edge[bend left] node[above, pos=0.3] {\(\EdgeValues{-1}{}\)} (v2)
              (v2) edge[bend left] node[below, pos=0.3]{\(\EdgeValues{+1}{0.5}\)} (v1)
              (v2) edge[bend left, pos=0.3] node [above] {\(\EdgeValues{0}{.5}\)} (v3)
              (v3) edge[bend left, pos=0.3] node [below] {\(\EdgeValues{0}{}\)} (v2)
        ;
    \end{tikzpicture}
    \caption{The \(\ObjectiveBWMP\)-value of the run \(\Run'\) is \(0\) but \(\Run'\) does not belong to \(\BWMP(0)\) . }
    \label{fig:BWMP}
\end{figure}
The example from~\cite{CDRR15} shown in Figure~\ref{fig:BWMP} illustrates that the converse of Proposition~\ref{prop:BWMP} does not hold in general.
For every run \(\Run\) in the MDP shown in the figure, there exists a window length \(\WindowLength\) such that \(\Run \in \FWMP(\WindowLength, -\epsilon)\).
For example, for \(\epsilon=0.01\), window length \(\WindowLength=100\) suffices.
Thus a \(-0.01\)-window that opens at \(v_0\) closes in at most \(100\) steps.
Note that for all \(\epsilon > 0\), the window length required is \(\frac{1}{\epsilon}\).
Now consider a run \(\Run'\) such that when the token takes the edge (\(v_1, v_2\))
for the \(i^{\text{th}}\) time, then it loops over \(v_2\) and \(v_3\) \(i\) times, after which the token moves to \(v_1\).
Clearly, for all \(\WindowLength \ge 1\), we have that \(\Run' \notin \FWMP(\WindowLength, 0)\) but for all \(\epsilon > 0\), there exists an \(\WindowLength\) such that \(\Run' \in \FWMP(\WindowLength, -\varepsilon)\).
Thus the \(\ObjectiveBWMP\)-value of \(\Run'\) is \(0\) while \(\Run' \notin \BWMP(0)\).

Note that both \(\FWMP_\MDP(\WindowLength, \Threshold)\) and \(\BWMP_\MDP(\Threshold)\) are Boolean prefix-independent objectives.
We omit the subscript \(\MDP\) when it is clear from the context.
As considered in previous works~\cite{CDRR15,BGR19,BDOR20}, the window length $\WindowLength$ is usually small (typically $\WindowLength \leq \abs{\Vertices}$), and therefore we assume that $\WindowLength$ is given in unary (while the edge-payoffs are given in binary).

\begin{comment}
\chcomment[id=SG]{Old}Given a window length \(\WindowLength \geq 1\) and a threshold \(\Threshold \in \Rationals\), a run \(\Run = \Vertex[0] \Vertex[1] \dotsm \)  in \(\Game\) satisfies the \emph{fixed window mean-payoff objective} \(\FWMP_{\Game}(\WindowLength, \Threshold)\) if from every position after some point, it is possible to start an infix of length at most $\WindowLength$ with mean payoff at least \(\Threshold\).

We first define Boolean versions of the objective: Given a window length~\(\WindowLength\) and a threshold~\(\Threshold\), a run~\(\Run\) satisfies the \emph{fixed window mean-payoff} objective, if, eventually in the run, every \(\Threshold\)-window closes in at most \(\WindowLength\) steps.
Formally,
...
Note that when \(\WindowLength = 1\), the \(\FWMPL\) objective is equivalent to a \Buchi\ objective.

Another objective considered in this work is the \emph{bounded window mean-payoff} objective.
In this objective, we require that there exist a window length \(\WindowLength\) such that the run satisfies the \(\FWMPL\) objective

The Boolean objectives \(\FWMPL\) and \(\BWMP\) can be generalized to quantitative objectives which we denote by \(\ObjectiveFWMPL\) and \(\ObjectiveBWMP\) respectively.
Given a window length \(\WindowLength\), the \(\ObjectiveFWMPL\)-value of a run \(\Run\) is equal to the supremum \(\Threshold\) such that \(\Run\) satisfies the \(\FWMPL\) objective.
When \(\WindowLength = 1\), this objective is equivalent to the liminf objective.
\end{comment}

\section{Problem definition}
\label{sec:bwc-bas-bpt-definitions}
We formally describe the notion of optimising expected \(\Objective\)-value with guarantees.
Given an MDP \(\MDP\), a vertex \(\Vertex\), a guarantee threshold \(\GuaranteeThreshold\), and an expectation threshold \(\ExpectationThreshold\), we consider the following three decision problems for optimising expectation while providing guarantees.
\begin{itemize}
    \item {\bf Beyond worst-case (\(\BWC\)) synthesis~\cite{BFRR17}} (Expectation maximisation with sure guarantee): 
    The problem here is to check if the supremum of \(\ExpectationOutcome{\Strategy}{\MDP, v}{\Objective}\) over all strategies \(\Strategy\) such that \(\Outcomev{v}{\MDP}{\Strategy} \subseteq \{\Objective \ge \GuaranteeThreshold \}\)
    (that is, all outcomes in \(\MDP\) starting from \(v\) that are consistent with \(\Strategy\) have \(\Objective\)-value at least \(\GuaranteeThreshold\))
    is at least \(\ExpectationThreshold\).
    We write this decision problem succinctly as
    \(v \models \BWCwArgs{\GuaranteeThreshold}{\ExpectationThreshold}\) in MDP \(\MDP\) for objective \(\Objective\).
    
    Note that for the \(\Outcomev{v}{\MDP}{\Strategy} \subseteq \{\Objective \ge \GuaranteeThreshold \}\) part, the probabilities are ignored and the environment is considered antagonistic in the sense that every play consistent with strategy \(\Strategy\) needs to satisfy the threshold Boolean constraint \(\{\Objective \ge \GuaranteeThreshold\}\).
    \item {\bf Beyond probability threshold (\(\BP\)) synthesis~\cite{CKK17}} (Expectation maximisation with probabilistic guarantee): Here we are given an additional probabilistic threshold \(p\).
    The problem here is to check if the supremum of \(\ExpectationOutcome{\Strategy}{\MDP, v}{\Objective}\) over all strategies \(\Strategy\) such that \(\Pr{\Strategy}{\MDP, v}{\{\Objective \ge \GuaranteeThreshold\}} \ge p\) is at least \(\ExpectationThreshold\).
    We write this decision problem succinctly as 
    \(v \models \BPwArgs{p}{\GuaranteeThreshold}{\ExpectationThreshold}\) in MDP \(\MDP\) for objective \(\Objective\).
    \item {\bf Beyond almost-sure (\(\BAS\)) synthesis~\cite{CR15}} (Expectation maximisation with almost-sure guarantee): This is also called beyond almost-sure (\(\BAS\)) synthesis~\cite{CR15}.
    The problem here is to check if the supremum of \(\ExpectationOutcome{\Strategy}{\MDP, v}{\Objective}\) over all strategies \(\Strategy\) such that \(\Pr{\Strategy}{\MDP, v}{\{\Objective \ge \GuaranteeThreshold\}} = 1\) is at least \(\ExpectationThreshold\).
    We write this decision problem succinctly as 
    \(v \models \BASwArgs{\GuaranteeThreshold}{\ExpectationThreshold}\) in MDP \(\MDP\) for objective \(\Objective\).
\end{itemize}
For each of these decision problems, we study the case where \(\Objective\) is either \(\ObjectiveFWMPL\) or \(\ObjectiveBWMP\). 
For the \(\BWC\) synthesis problem, if the answer is yes, then for every \(\epsilon >0\), we construct a strategy that achieves an expected \(\Objective\)-value of at least \(\ExpectationThreshold - \epsilon\). 
For the \(\BP\) and the \(\BAS\) synthesis problems,
if the answer is yes, then we construct strategies that achieve the expected \(\Objective\)-value of at least \(\ExpectationThreshold\) as well as the specified guarantees.

For all three synthesis problems for the classical mean-payoff objective and for the window mean-payoff objectives, we can assume without loss of generality that the guarantee threshold \(\GuaranteeThreshold\) is \(0\).
This is because we have \(v \models \BWC(\GuaranteeThreshold, \ExpectationThreshold)\) in an MDP \(\MDP\) if and only if \(v \models \BWC(0, \ExpectationThreshold - \GuaranteeThreshold)\) in a new MDP \(\MDP_{-\GuaranteeThreshold}\) (obtained from \(\MDP\) by subtracting \(\GuaranteeThreshold\) from every edge payoff in \(\MDP\)).
Similarly, we can set \(\GuaranteeThreshold = 0\) for \(\BP\) and \(\BAS\) without loss of generality.

\section{Expected fixed window mean-payoff value with guarantees}%
\label{sec:fwmp}
In this section, we show that the \(\BWC\), \(\BP\), and the \(\BAS\) synthesis problems for \(\ObjectiveFWMPL\) can be solved with no additional complexity than that of either of the special cases: maximising the expectation without any guarantee, or ensuring guarantee surely, almost-surely, or with certain probability while disregarding any expected performance.

\subsection{Sure guarantee}%
\label{sub:bwc-fwmp}
In this section, we give an algorithm to solve the \(\BWC\) synthesis problem for the \(\ObjectiveFWMPL\) objective.
Recall that the expected \(\ObjectiveFWMPL\)-value of a vertex \(v\) is defined as the supremum of the expected \(\ObjectiveFWMPL\)-values \(\ExpectationOutcome{\Strategy}{\MDP, v}{\ObjectiveFWMPL}\) over all strategies \(\Strategy\) of the player.
We show in \Cref{ex:bwc-optimal-strategy-may-not-exist} that in general, an optimal strategy achieving the expected \(\ObjectiveFWMPL\)-value \(\ExpectationThreshold\) while also ensuring the sure guarantee of \(0\) may not exist, but for every \(\epsilon > 0\), an \(\epsilon\)-optimal strategy exists.
We thus show how to construct \(\epsilon\)-optimal strategies for \(\BWC(0, \ExpectationThreshold)\) satisfaction.

\begin{example}%
\label{ex:bwc-optimal-strategy-may-not-exist}
    Consider once again the MDP shown in \Cref{fig:bwc-example}.
    We want to determine if \(v_{2} \models \BWCwArgs{0}{2}\) for the \(\ObjectiveFWMPL\) objective for window length \(\WindowLength = 3\).
    If the token somehow reaches \(v_4\), then from there, the player has a strategy to ensure that the \(\ObjectiveFWMPL\)-value of the outcome is surely \(2\), and thus, the player can ensure that the expected \(\ObjectiveFWMPL\)-value from \(v_4\) is at least \(2\) as well.
    Note that for every successive visit of the token to \(v_{7}\), the player has to alternate between \(v_{6}\) and \(v_{8}\) to ensure that the outcome is \(2\).
    However, starting from \(v_{2}\), the player does not have a strategy to reach \(v_{4}\) surely.
    This is because there is an outcome \((v_2 v_3)^{\omega}\) that does not get the token to reach \(v_{4}\).
    If the token remains in the set \(\{v_0, v_1, v_2\}\), then the \(\ObjectiveFWMPL\)-value that can be surely attained is \(0\).
    However, if the player tries to move the token from \(v_2\) to \(v_3\) some fixed (but large) number \(N\) of times, then the probability of reaching \(v_4\) can be made close to \(1\). 
    If after \(N\) tries, the token does not reach \(v_{4}\), then the player can choose to keep the token in the set \(\{v_{0}, v_{1}, v_{2}\}\) and thus surely get a \(\ObjectiveFWMPL\)-value of  \(0\).
    This gives a strategy that surely ensures that the \(\ObjectiveFWMPL\)-value of the outcome is non-negative and the expected \(\ObjectiveFWMPL\)-value is at least \(2 - \epsilon\) for all \(\epsilon > 0\).
\end{example}

We present an algorithm (\Cref{alg:wmp-sure-guarantee}) that, given a vertex \(\VertexInitial\) in an MDP \(\InputMDP\), and a threshold \(\ExpectationThreshold\), decides if \(\VertexInitial \models \BWCwArgs{0}{\ExpectationThreshold}\).
We give an intuitive description of \Cref{alg:wmp-sure-guarantee} and prove its correctness. 
We also show that it runs in time that is polynomial in the size of the input, and explain how it yields an \(\epsilon\)-optimal strategy for the player.
\begin{algorithm}
\caption{\(\FWMPL\) objective with sure guarantee while maximising expectation}\label{alg:wmp-sure-guarantee}
    \begin{algorithmic}[1]
        \Require MDP \(\InputMDP\), vertex \(\VertexInitial \in V\), window length \(\WindowLength\),  and expectation threshold \(\ExpectationThreshold\)
        \Ensure {\Yes} if and only \(\VertexInitial \models \BWC(0, \ExpectationThreshold)\)
        \State Compute \(\SureWinningSet{\FWMPL}{0}\), the sure winning region in \(\InputMDP\) for objective  \(\ThresholdObjectiveFWMPL{\ge 0}\).%
        \label{alg-line:sure-sure-winning-region-computation}
        \If{\(\VertexInitial \notin \SureWinningSet{\FWMPL}{0}\)}
            \State \Return {\No} 
        \EndIf
        \If{\(\ExpectationThreshold \le 0\)} 
            \State \Return {\Yes}
        \EndIf
        \State Construct \(\PrunedMDP \define \InputMDP \restriction \SureWinningSet{\FWMPL}{0}\), the MDP obtained by restricting \(\InputMDP \) to \(\SureWinningSet{\FWMPL}{0}\).
        \For{\(v \in \VerticesPlayer\) in \(\PrunedMDP\)}
            \State Compute the maximum \(\SureValue{v}\) such that \(v\) belongs to the sure winning region in \(\PrunedMDP\) for objective \(\ThresholdObjectiveFWMPL{\ge \SureValue{v}}\).%
            \label{alg-line:sure-sure-value-computation}
        \EndFor
        \State Construct \(\PrunedMDPWithSelfLoops\) from \(\PrunedMDP\) as follows: \newline
        Change payoffs of all edges to \(-1\). \newline
        For each player vertex \(v\), add a self-loop with payoff \(\SureValue{v}\). 
        \State \Return {\Yes} if and only if \(\ExpectationValue{\PrunedMDPWithSelfLoops}{\VertexInitial}{\ObjectiveMP} \ge \ExpectationThreshold\).\label{alg-line:returnBWC}
    \end{algorithmic}
\end{algorithm}

\paragraph*{Description of \Cref{alg:wmp-sure-guarantee}.}
We first compute the sure winning region, \(\SureWinningSet{\FWMPL}{0}\), of the player for the threshold objective \(\ThresholdObjectiveFWMPL{\ge 0}\) (Line~\ref{alg-line:sure-sure-winning-region-computation}). 
The sure winning region of the player in MDP \(\InputMDP\) is the same as the winning region of the player in the adversarial two-player game \(\Game_{\InputMDP}\) obtained by viewing probabilistic vertices of \(\InputMDP\) as vertices of an adversary.
We compute the winning region of player in \(\Game_{\InputMDP}\) using \cite[Algorithm~1]{CDRR15}.
If the vertex \(\VertexInitial\) does not belong to \(\SureWinningSet{\FWMPL}{0}\), then the sure guarantee cannot be satisfied from \(\VertexInitial\), and we have that \(\VertexInitial \not\models \BWC(0, \ExpectationThreshold)\) for all \(\ExpectationThreshold \in \Rationals\), and the algorithm returns {\No}.

Otherwise, we have that \(\VertexInitial \in \SureWinningSet{\FWMPL}{0}\), and thus, there exists a strategy that ensures the sure satisfaction of \(\ThresholdObjectiveFWMPL{\ge 0}\) from \(\VertexInitial\).
We now check if \(\ExpectationThreshold \le 0\).
This is because the sure satisfaction of \(\ThresholdObjectiveFWMPL{\ge 0}\) from \(\VertexInitial\) implies that the expected \(\ObjectiveFWMPL\)-value of \(\VertexInitial\) is at least \(0\), and in particular, if \(\ExpectationThreshold \le 0\), then it follows that \(\VertexInitial \models \BWC(0, \ExpectationThreshold) \), and the algorithm returns {\Yes}.

Finally, we arrive at the interesting case that is  \(\ExpectationThreshold > 0\).
From every vertex \(v\) in \(\SureWinningSet{\FWMPL}{0}\), 
there exists a sure winning strategy for objective \(\ThresholdObjectiveFWMPL{\ge 0}\), 
and every such sure winning strategy never moves the token out of \(\SureWinningSet{\FWMPL}{0}\).
Thus, we can prune all vertices from \(\InputMDP\) that are not in \(\SureWinningSet{\FWMPL}{0}\) to obtain \(\PrunedMDP\), and 
we have that the sets of sure winning strategies for \(\ThresholdObjectiveFWMPL{\ge 0}\) in \(\InputMDP\) and \(\PrunedMDP\) are the same.
Next, in Line~\ref{alg-line:sure-sure-value-computation}, for each player vertex \(v\) in \(\PrunedMDP\), we compute \(\SureValue{v}\), that is the maximum \(\ObjectiveFWMPL\)-value that the player can surely ensure in \(\PrunedMDP\) starting from \(v\).
The sure \(\ObjectiveFWMPL\)-value of each vertex can be computed in polynomial time using binary search (details appear later).
Using \(\PrunedMDP\) and the sure-values \(\SureValue{v}\), we construct a new MDP \(\PrunedMDPWithSelfLoops\) as follows.
\begin{itemize}
    \item  The set of vertices of \(\PrunedMDPWithSelfLoops\) is the same as the set of vertices as \(\PrunedMDP\).
    \item For each edge in \(\PrunedMDP\), we have the same edge present in \(\PrunedMDPWithSelfLoops\), but with payoff \(-1\). 
    In addition, for each player vertex \(v\) in \(\PrunedMDPWithSelfLoops\), we add a self-loop with payoff \(\SureValue{v}\) to \(v\).
\end{itemize}
We compute the expected \(\ObjectiveMP\)-value of \(\VertexInitial\) in \(\PrunedMDPWithSelfLoops\) using linear programming~\cite{Puterman94}.
The correctness of Line~\ref{alg-line:returnBWC} follows from Lemma~\ref{lem:sure-reduction-to-mean-payoff}.

\begin{lemma}%
\label{lem:sure-reduction-to-mean-payoff}
    For all \(\gamma \ge 0\), we have that \(\VertexInitial \models \BWC(0, \gamma)\) in \(\InputMDP\) if and only if  \(\ExpectationValue{\PrunedMDPWithSelfLoops}{\VertexInitial}{\ObjectiveMP} \ge \gamma\).
\end{lemma}

\begin{proof}
    First, we show for all \(\gamma \ge 0\) that \(\ExpectationValue{\PrunedMDPWithSelfLoops}{\VertexInitial}{\ObjectiveMP} \ge \gamma\) implies \(\VertexInitial \models \BWC(0, \gamma)\) in \(\InputMDP\).
    Let \(\StrategyMP\) be a deterministic memoryless optimal strategy from \(\VertexInitial\) in \(\PrunedMDPWithSelfLoops\) for the expectation of \(\ObjectiveMP\) (the existence of such a strategy follows from~\cite{Puterman94}).
    Each outcome of \(\StrategyMP\) from \(\VertexInitial\) is almost-surely a run of the form \(\VertexInitial \cdot v_1 \cdot v_2 \cdots v_m \cdot u^{\omega}\) for some \(m \ge 0\) for some vertices \(v_1, v_2, \ldots, v_m, u \in \PrunedMDPWithSelfLoops\).
    That is, the run begins from \(\VertexInitial\) and almost-surely eventually reaches a vertex \(u\) from which it always takes the self-loop on \(u\) with edge payoff \(\SureValue{u}\).
    This is because \(\StrategyMP\) is memoryless and the payoff \(\SureValue{u}\) of the self-loop for every vertex \(u\) is non-negative, while all other edges in \(\PrunedMDPWithSelfLoops\) have a negative payoff of \(-1\).
    For an outcome of the strategy \(\StrategyMP\) from \(\VertexInitial\) in \(\PrunedMDPWithSelfLoops\), let \(\{u_1, u_2, \ldots, u_k\}\) denote the set of all the vertices that the token reaches with positive probabilities following which the token takes the self-loop forever.
    Further, for \(1 \le i \le k\), let \(p_{i}\) denote the probability that an outcome of \(\StrategyMP\) reaches \(u_{i}\).
    Since the token almost-surely reaches one of the vertices in \(\{u_1, \ldots, u_{k}\}\), we have that \(\sum_{i=1}^k p_{i} = 1\).
    Moreover, if in an outcome \(\Run\), the token loops on vertex \(u_{i}\) for some \(1 \le i \le k\), then the \(\ObjectiveMP\)-value of \(\Run\) is \(\SureValue{u_{i}}\).
    Thus, the expected \(\ObjectiveMP\)-value of \(\VertexInitial\) in \(\PrunedMDPWithSelfLoops\) is \(\sum_{i=1}^k p_i \cdot \SureValue{u_{i}}\), which is at least \(\gamma\) by hypothesis.
    
    To show that \(\VertexInitial \models \BWC(0, \gamma)\) in \(\InputMDP\), we show a construction, for every \(\epsilon > 0\), of a strategy \(\Strategy^{*}_{\epsilon}\) such that outcomes of this strategy from \(\VertexInitial\) surely satisfy \(\ThresholdObjectiveFWMPL{\ge 0}\) and we also have \(\ExpectationOutcome{\Strategy^{*}_{\epsilon}}{\MDP, \VertexInitial}{\ObjectiveFWMPL} \ge \gamma - \epsilon\). 
    We construct \(\Strategy^{*}_{\epsilon}\) by combining strategies \(\StrategyMP\) and \(\StrategySure^{\FWMPL}\), the optimal strategies for expected \(\ObjectiveMP\)-value in \(\PrunedMDPWithSelfLoops\) and the sure-winning strategy for objective \(\ThresholdObjectiveFWMPL{\ge \SureValue{v}}\) in \(\InputMDP\) respectively.
    From \Cref{lem:optimalreach}, for every \(\epsilon > 0\), we can choose a large enough \(N\) such that the token reaches each \(u_i\) with probability at least \(p_i - \epsilon / (\abs{\Vertices} \cdot W)\), where \(W\) is the maximum edge payoff appearing in \(\InputMDP\).
    Let \(\epsilon' = \epsilon / (\abs{\Vertices} \cdot W)\) for brevity.
    The strategy \(\Strategy_{\epsilon}^{*}\) mimics \(\StrategyMP\) until \(\StrategyMP\) starts looping or until \(N\) steps have passed, whichever comes first.  
    Then, if the token is on some vertex \(v\), then \(\Strategy^{*}_{\epsilon}\) switches to mimicking the sure-winning strategy \(\StrategySure^{\FWMPL}\) for the rest of the run.
    Thus, when this strategy \(\Strategy^{*}_{\epsilon}\) switches its mode, the token is on vertices \(u_{1}, u_{2}, \ldots, u_{k}\) with probabilities at least \(p_{1} - \epsilon', p_{2} - \epsilon', \ldots, p_{k} - \epsilon'\) respectively.
    It follows that during the switch, the probability that the token would be on neither of these vertices, that is, the token would be on a vertex in \(\SureWinningSet{\FWMPL}{0} \setminus \{u_{1}, \cdots, u_{k}\}\) is at most \(k \cdot \epsilon'\).
    Since the sure value of every vertex in \(\SureWinningSet{\FWMPL}{0}\) is non-negative and at most \(W\), we have that the expected \(\ObjectiveFWMPL\)-value of an outcome of \(\Strategy^{*}_{\epsilon}\) is at least \(\sum_{i=1}^{k} \SureValue{u_i} \cdot (p_i - \epsilon') + 0 \cdot (1 - \sum_{i=1}^{k} (p_i - \epsilon'))\ge \gamma - \sum_{i = 1}^{k}\SureValue{u_{i}} \cdot \epsilon' \ge \gamma - k \cdot W \cdot \epsilon' \ge \gamma - \epsilon\).
    The \(\ObjectiveFWMPL\)-value of an outcome of this strategy is surely non-negative since \(\SureValue{v} \ge 0\) for all vertices \(v\) in \(\PrunedMDP\).
    Thus, we have that \(\VertexInitial \models \BWC(0, \gamma)\).
    
    Now, we show the converse, that is, we show for all \(\gamma \ge 0\) that  \(\VertexInitial \models \BWC(0, \gamma)\) in \(\InputMDP\) implies  \(\ExpectationValue{\PrunedMDPWithSelfLoops}{\VertexInitial}{\ObjectiveMP} \ge \gamma\).
    Suppose that from vertex \(\VertexInitial\), the player has an \(\epsilon\)-optimal strategy \(\Strategy^{*}_{\epsilon}\) for \(\BWC(0, \gamma)\) from \(\VertexInitial\) in \(\InputMDP\).
    Using \(\Strategy^{*}_{\epsilon}\), we describe a strategy \(\StrategyMP\) from \(\VertexInitial\) in \(\PrunedMDPWithSelfLoops\) that achieves expected \(\ObjectiveMP\)-value at least \(\gamma\), that is, \(\ExpectationOutcome{\StrategyMP}{\PrunedMDPWithSelfLoops, \VertexInitial}{\ObjectiveMP} \ge \gamma\).
    From \Cref{lem:longrunMEC}, we have that an outcome of \(\Strategy^{*}_{\epsilon}\) almost-surely eventually reaches and stays in a MEC from which it never exits. 
    Suppose that starting from \(\VertexInitial\), an outcome of the strategy \(\Strategy^{*}_{\epsilon}\) ends up in MECs \(\MEC_1, \MEC_2, \ldots, \MEC_k\) with probability \(p_{1}, p_{2}, \ldots, p_{k}\) respectively.
    If in an outcome \(\Run\) of \(\Strategy^{*}_{\epsilon}\), the token reaches the MEC \(\MEC_{i}\) and never leaves, then the \(\ObjectiveFWMPL\)-value of \(\Run\) is at most \(\max \{\SureValue{v} \mid v \in \MEC_{i}\}\), and we denote this by \(\SureValue{\MEC_{i}}\).
    Thus, the expected \(\ObjectiveFWMPL\)-value of an outcome of \(\Strategy^{*}_{\epsilon}\) is at least \(\sum_{i=1}^{k} \SureValue{\MEC_{i}} \cdot p_{i}\), that is, \(\sum_{i=1}^{k} \SureValue{\MEC_{i}} \cdot p_{i} \ge \gamma - \epsilon\).
    The strategy \(\StrategyMP\) mimics \(\Strategy^{*}_{\epsilon}\) to almost-surely reach the same MECs in \(\PrunedMDPWithSelfLoops\) with the same probabilities. 
    Now, in each MEC \(\MEC\) in \(\PrunedMDPWithSelfLoops\), the strategy \(\StrategyMP\) can ensure expected \(\ObjectiveMP\)-value \(\max\{ \SureValue{v} \mid v \in \MEC\} \) by almost-surely reaching the vertex \(v\) in \(\MEC\) with the maximum \(\SureValue{v}\), and then looping on \(v\) forever. 
    Thus, we have that by following this strategy \(\StrategyMP\), the expected \(\ObjectiveMP\)-value of an outcome from \(\VertexInitial\) is at least \(\sum_{i=1}^{k} \SureValue{\MEC_{i}} \cdot p_{i}\), and thus, we have that \(\ExpectationValue{\PrunedMDPWithSelfLoops}{\VertexInitial}{\ObjectiveMP} \ge \gamma-\epsilon\).
    Since this holds for every \(\epsilon\), we have that \(\ExpectationValue{\PrunedMDPWithSelfLoops}{\VertexInitial}{\ObjectiveMP} \ge \gamma\).
\end{proof}

\paragraph*{Memory requirement for \(\epsilon\)-optimal strategies.}
In the proof above, we describe a strategy \(\Strategy^{*}_{\epsilon}\) in \(\InputMDP\) that is an \(\epsilon\)-optimal strategy for \(\BWC(0, \ExpectationThreshold)\) constructed from the optimal strategy \(\StrategyMP\) in \(\PrunedMDPWithSelfLoops\) and the optimal sure strategy \(\StrategySure^{\FWMPL}\) in \(\InputMDP\).
Recall that \(\StrategyMP\) is memoryless, and that in each run in \(\PrunedMDPWithSelfLoops\) consistent with \(\StrategyMP\), the token almost-surely eventually reaches a vertex on which it loops.
Given \(\epsilon > 0\), the \(\epsilon\)-optimal strategy \(\Strategy^{*}_{\epsilon}\) mimics \(\StrategyMP\) until a looping vertex is reached or until \(N\) steps have passed (where \(N\) depends on \(\epsilon\)), after which, \(\Strategy^{*}_{\epsilon}\) switches to mimicking \(\StrategySure^{\FWMPL}\) for the rest of the play.
Since \(\StrategyMP\) is memoryless, and
 \(\StrategySure^{\FWMPL}\) requires at most \(\WindowLength\) memory~\cite{CDRR15},
the memory required by \(\Strategy^{*}_{\epsilon}\) is at most \(\max\{N, \WindowLength\}\).
Thus, deterministic finite memory strategies suffice for \(\BWC\).

As \(N\) increases, the probability that the token reaches a looping vertex in \(N\) steps increases.
Given \(\epsilon > 0\), we want an upper bound on the memory requirements of \(\epsilon\)-optimal strategies, and thus, we want to find the smallest \(N\) such that the token reaches a looping vertex in \(N\) steps with probability at least \(1 - \epsilon\).
\Cref{ex:epsilon-N-bound} shows a conservative MDP in which the value of \(N\) is logarithmic in \(1/\epsilon\). 
The MDP is conservative since in order to reach a looping vertex \(u_m\) from the initial vertex \(u_0\), the ``correct'' transition must be taken at probabilistic vertices \(m\) times in a row.
Any time the token makes a ``wrong'' transition, it is sent back to the initial vertex, causing it to lose all progress towards reaching the looping vertex.

\begin{example}%
\label{ex:epsilon-N-bound}
    Consider an MDP \(\MDP_{m}\) (where \(m\) is a positive integer) as shown in \Cref{fig:N-epsilon-bound}.
    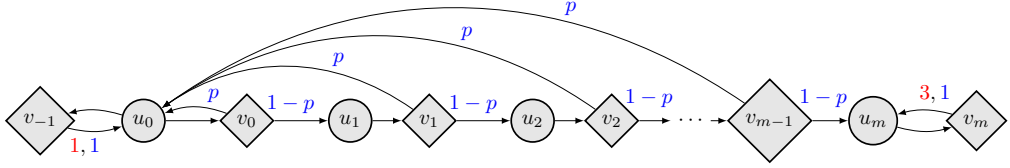
\begin{figure}[t]
        \centering
        \scalebox{0.8}{
            \begin{tikzpicture}
                \node[state] (u0) {\(u_{0}\)};
                \node[random, draw, left of=u0] (vminus1) {\(v_{-1}\)};
                \node[random, draw, right of=u0] (v0) {\(v_{0}\)};
                \node[state, right of=v0] (u1) {\(u_{1}\)};
                \node[random, draw, right of=u1, xshift=-4mm] (v1) {\(v_{1}\)};
                \node[state, right of=v1] (u2) {\(u_{2}\)};
                \node[random, draw, right of=u2, xshift=-4mm] (v2) {\(v_{2}\)};
                \node[right of=v2, xshift=-4mm] (dots) {\(\cdots\)};
                \node[random, draw, right of=dots, xshift=-4mm] (vk) {\(v_{m-1}\)};
                \node[state, right of=vk] (uV) {\(u_{m}\)};
                \node[random, draw, right of=uV] (vV) {\(v_{m}\)};
                \draw 
                      (u0) edge (v0)
                      (v0) edge[bend right=20] node[above, pos=0.2]{\(\color{blue}{p}\)} (u0)
                      (v0) edge node[above, pos=0.3]{\(\color{blue}{1-p}\)} (u1)
                      
                      (u1) edge (v1)
                      (v1) edge[bend right=35] node[above, pos=0.3]{\(\color{blue}{p}\)} (u0)
                      (v1) edge node[above, pos=0.3]{\(\color{blue}{1-p}\)} (u2)
                     
                      (u2) edge (v2)
                      (v2) edge[bend right=35] node[above, pos=0.3]{\(\color{blue}{p}\)} (u0)
                      (v2) edge node[above, pos=0.3, yshift=1mm]{\(\color{blue}{1-p}\)} (dots)
    
                      (dots) edge (vk)
                      
                      (vk) edge node[above, pos=0.3, yshift=1mm]{\(\color{blue}{1-p}\)} (uV)
                      (vk) edge[bend right=35] node[above, pos=0.3]{\(\color{blue}{p}\)} (u0)
                      
                      (uV) edge[bend right=15] (vV)
                      (vV) edge[bend right=15] node[above, pos=0.3]{\(\EdgeValues{3}{1}\)} (uV)
                      
                      (u0) edge[bend right=15] (vminus1)
                      (vminus1) edge[bend right=15] node[below, pos=0.3]{\(\EdgeValues{1}{1}\)} (u0)
                ;
            \end{tikzpicture}
            }
        \caption{The edge \((v_{-1}, u_0)\) has payoff \(1\) and edge \((v_m, u_m)\) has payoff \(3\). Every other edge has payoff \(-1\).}
        \label{fig:N-epsilon-bound}
    \end{figure}
    The MDP \(\MDP_{m}\) consists of \(2m + 3\) vertices, of which \(m+1\) are player vertices \(u_0, u_1, \ldots, u_m \in \VerticesPlayer\) and \(m+2\) are probabilistic vertices \(v_{-1}, v_0, v_1, \ldots, v_m \in \VerticesRandom\). 
    The edges in \(\MDP_{m}\) are as follows:
    \begin{itemize}
        \item 
        The player vertex \(u_0\) has two out-neighbours \(v_{-1}\) and \(v_0\). Both edges \((u_0, v_{-1})\) and \((u_0, v_0)\) have payoff \(-1\).
        \item 
        For each \(1 \le i \le m\), the player vertex \(u_i\) has one out-neighbour \(v_i\), and the edge \((u_i, v_i)\) has payoff \(-1\).
        \item 
        The probabilistic vertex \(v_{-1}\) has one out-neighbour \(u_0\).
        The edge \((v_{-1}, u_0)\) has payoff \(+1\).
        \item 
        For each \(0 \le i \le m - 1\), the probabilistic vertex \(v_i\) has a ``right'' out-neighbour \(u_{i+1}\) and a ``reset'' out-neighbour \(u_{0}\). 
        The ``right'' edge \((v_i, u_{i+1})\) occurs with probability \(1-p\) and the ``reset'' edge \((v_i, u_0)\) occurs with probability \(p\).
        Both edges have payoff \(-1\).
        \item 
        The probabilistic vertex \(v_m\) has one out-neighbour \(u_m\).
        The edge \((v_m, u_m)\) has payoff \(+3\).
    \end{itemize}
    Note that starting from \(u_0\), the token reaches \(u_m\) if and only if the player moves the token to \(v_0\) and the token then goes ``right'' at probabilistic vertices \(m\) times in a row, which happens with probability \((1-p)^m\) (which is a constant positive number).
    Thus, if the player always moves the token to \(v_0\) each time the token is on \(u_0\),  then with probability \(1\) the token eventually reaches \(u_m\) and then cycles between \(u_m\) and \(v_m\) forever. 
    Hence, with probability \(1\), the \(\ObjectiveFWMPL\)-value of an outcome of the ``always \(u_0 \to v_0\)'' strategy is \((-1 + 3)/2 = 1\) for all \(\WindowLength \ge 2\).
    However, this strategy does not ensure that the token surely eventually reaches \(u_m\).
    Indeed, a run consistent with this strategy is one that cycles between \(u_0\) and \(v_0\) forever, and this run has \(\ObjectiveFWMPL\)-value of \(-1\) for \(\WindowLength \ge 1\).
    Thus, this strategy is not optimal for \(\BWCwArgs{0}{\ExpectationThreshold}\) for any threshold \(\ExpectationThreshold\).
    On the other hand, if the player always moves the token from \(u_1\) to \(v_{-1}\), then the token loops between \(u_1\) and \(v_{-1}\) forever, and the sure \(\ObjectiveFWMPL\)-value of an outcome of this strategy is \((-1 + 1)/2 = 0\) for all \(\WindowLength \ge 2\). 
    This is an optimal strategy for \(\BWCwArgs{0}{0}\). 

    With the help of memory, the player can achieve a better \(\ObjectiveFWMPL\)-value in expectation while also guaranteeing a \(\ObjectiveFWMPL\)-value of \(0\) surely.
    Using the \(\epsilon\)-optimal strategy described in the proof of \Cref{lem:sure-reduction-to-mean-payoff}, there exists an \(\epsilon\)-optimal strategy for \(\BWCwArgs{0}{1}\).
    The strategy is the following:
    Let \(N\) be a positive integer.
    As long as \(N\) steps of the player have not elapsed since the beginning of the game, whenever the token reaches \(u_0\), take the \((u_0, v_0)\) edge.
    Otherwise, if at least \(N\) steps of the player have elapsed, then whenever the token reaches \(u_0\), always take the \((u_0, v_{-1})\) edge.
    Every outcome of this strategy either loops between \(u_m\) and \(v_m\) (if the token reaches \(u_m\) in \(N\) steps) or loops between \(u_0\) and \(v_{-1}\) (if the token does not reach \(u_m\) in \(N\) steps).
    Thus, the sure \(\ObjectiveFWMPL\)-value of an outcome of this strategy is at least \(0\), and moreover, as \(N\) grows, the probability that the token reaches \(u_m\) in \(N\) steps approaches \(1\), and therefore, the expected \(\ObjectiveFWMPL\)-value of an outcome of this strategy also approaches \(1\) for all \(\WindowLength \ge 2\).

    The size of the memory of such an \(\epsilon\)-optimal strategy is \(\max\{N, \WindowLength\}\).
    Thus, for a given \(\epsilon\), we want to find the smallest \(N\) such that the probability of reaching \(u_{m}\) from \(u_0\) in at most \(N\) steps is at least \(1-\epsilon\).
    The problem reduces to the following: 
    Let a coin show heads with probability \(p\) and tails with probability \(1-p\). Then, find the smallest positive integer \(N\) such that the probability of seeing at least \(m\) consecutive tails in \(N\) coin tosses is at least \(1-\epsilon\). 
    Equivalently, find the smallest positive integer \(N\) such that the probability of ``no \(m\) consecutive tails in \(N\) tosses'' is at most \(\epsilon\).
    Let \(T_N\) denote this probability of ``no \(m\) consecutive tails in \(N\) tosses''.

    If \(N < m\), then with probability \(1\) there will be no sequence of \(m\) consecutive tails in \(N\) coin tosses. 
    Thus, we have that \(T_0 = T_1 = \dots = T_{m-1} = 1\).
    Otherwise, if \(N \ge m\), then in order for there to be no sequence of \(m\) consecutive tails, there must be at least \(1\) head in the first \(m\) tosses.
    We can characterise sequences of \(N\) coin tosses not containing \(m\) consecutive tails based on the first time we see a head.
    For \(1 \le i \le m\), the probability that the first head appears on the \(i^{\text{th}}\) toss, and then no sequence of \(m\) consecutive tails appears in the remaining \(N - i\) tosses, is \((1-p)^{i-1} \cdot p \cdot T_{N - i}\).
    Thus, the following homogeneous linear recurrence relation solves for \(T_N\).
    \[T_N = p \cdot T_{N-1}+(1-p) \cdot p \cdot T_{N-2} + \dots + (1-p)^{m-1} \cdot p \cdot T_{N-m}, \quad T_0 = T_1 = \dots = T_{m-1} = 1\]
    Solving this recurrence relation gives \(T_N = a_1 c_1^N + a_2 c_2^N + \dots + a_m c_m^N\) for \(c_i < 1\) and constants \(a_i\). 
    Thus, the smallest \(N\) such that \(T_{N} \le \varepsilon\) is logarithmic in \(1/\epsilon\). 
    \qed
\end{example}

\paragraph*{Running time analysis.}
Given a vertex \(v\) in an MDP \(\InputMDP\) and a threshold \(\GuaranteeThreshold\),
the problem of determining if the player has a strategy to surely satisfy the objective \(\ThresholdObjectiveFWMPL{\ge \GuaranteeThreshold}\) from \(v\) is in polynomial time~\cite{CDRR15}.
Thus, the set \(\SureWinningSet{\FWMPL}{0}\) of all such vertices can be computed in polynomial time. 

Computing the sure value \(\SureValue{v}\) for the \(\ObjectiveFWMPL\) for each vertex can be done in polynomial time by using binary search.
Since \(\SureValue{v}\) is non-negative, recall that we can write \(\SureValue{v}\) as a fraction \(a/b\), where \(b \in \{1, \ldots, \WindowLength\}\) and \(a \in \{0, 1, \ldots, W \cdot \WindowLength\} \).
Thus, there are at most \(W \cdot \WindowLength^{2}\) different values that \(\SureValue{v}\) can take.
This value can be found using binary search, where for each possible value \(\GuaranteeThreshold\), we check if \(v\) belongs to the sure winning region in \(\PrunedMDP\) for the threshold objective \(\ThresholdObjectiveFWMPL{\ge \GuaranteeThreshold}\).
Since \(\SureValue{v}\) takes at most \(W \cdot \WindowLength^{2}\) different values, it takes at most \(\log (W \cdot \WindowLength^{2})\) checks to arrive at \(\SureValue{v}\).
Since \(W\) is given in binary and \(\WindowLength\) in unary, we have that \(\SureValue{v}\) can be computed in time that is polynomial in the size of the input.

Expectation of \(\ObjectiveMP\) objective can also be solved in polynomial time using linear programming~\cite{Puterman94}.
Thus, \Cref{alg:wmp-sure-guarantee} runs in polynomial time. 
We summarise the results in the following theorem.
\begin{theorem}%
\label{thm:fwmp-sure-result}
    The \(\BWC\) synthesis problem for the \(\ObjectiveFWMPL\) objective is in \(\PTime\), and if \(v \models \BWC(\GuaranteeThreshold, \ExpectationThreshold)\), then for every \(\epsilon > 0\), there exists an \(\epsilon\)-optimal finite-memory deterministic strategy from \(v\).
\end{theorem}

\subsection{Probabilistic guarantee}%
\label{sec:bp-fwmp}
Next, we look at the \(\BP((p, \GuaranteeThreshold), \ExpectationThreshold)\) problem, which generalises the \(\BAS\) problem. 
Given an MDP \(\InputMDP\) and a vertex \(\VertexInitial\), we want to determine if there exists a strategy \(\Strategy\) of the player from \(\VertexInitial\) such that the outcome simultaneously satisfies \(\ThresholdObjectiveFWMPL{\ge \GuaranteeThreshold}\) with probability at least \(p\) and \(\ExpectationOutcome{\Strategy}{\InputMDP, \VertexInitial}{\ObjectiveFWMPL} \ge \ExpectationThreshold\).
As before, we assume without loss of generality that \(\GuaranteeThreshold\) is equal to zero.

In contrast to \(\BWC\), in the case of \(\BP\), the problem is interesting even when \(\ExpectationThreshold \le 0\). 
This is because satisfying the threshold objective \(\ThresholdObjectiveFWMPL{\ge 0}\) with probability at least \(p\) does not necessarily imply that the expectation is at least \(\ExpectationThreshold\), even when \(\ExpectationThreshold \le 0\).
Further, unlike in the case of \(\BWC\), we cannot prune away the set of vertices from which the player cannot satisfy \(\ThresholdObjectiveFWMPL{\ge 0}\) with probability at least \(p\).
This is because, in trying to satisfy the expectation threshold, the token may end up visiting a vertex from which the probability of satisfying \(\ThresholdObjectiveFWMPL{\ge 0}\) is less than \(p\). 
Thus, in the case of \(\BP\), in \Cref{alg:wmp-probabilistic-guarantee}, we do not prune the MDP \(\InputMDP\).

\begin{example}%
\label{ex:bp-example}
    In \Cref{fig:bp-example}, we see an MDP \(\MDP\) in which we want to determine if \(v_{3} \models \BPwArgs{0.5}{0}{2}\) for \(\ObjectiveFWMPL\) for window length \(\WindowLength = 2\).
    If the player keeps the token in the MEC consisting of \(\{v_{0}, v_{1}, v_{2}, v_{3}\}\) forever with probability \(1\), then a \(\ObjectiveFWMPL\)-value of \(+1\) (which is non-negative) is ensured with probability \(1\), which satisfies the guarantee threshold.
    However, this gives only \(1\) in expectation, which is not sufficient to satisfy the expectation threshold \(2\).
    On the other hand, if the player moves the token from \(v_{3}\) to \(v_{4}\) with probability \(1\), then the token reaches the MEC \(\{v_5, v_{7}\}\) with probability \(0.6\) and achieves a value of \(-1\), whereas the token reaches the MEC \(\{v_{6}, v_{8}\}\) with probability \(0.4\) which has a value of \(9\).
    Thus, from \(v_4\), the expected value is \(0.6 \cdot (-1) + 0.4 \cdot 9 = 3\) and thus the expectation threshold is satisfied. 
    However, the token achieves non-negative value with probability only \(0.4\), and the guarantee threshold is not satisfied with sufficient probability.
    On the other hand, if the player chooses a randomised strategy from \(v_3\), that is, if she chooses to stay inside \(\{v_{0}, v_{1}, v_{2}, v_{3} \}\) with probability \(\frac{1}{4}\) and to go to \(v_4\) with probability \(\frac{3}{4}\), then not only is the guarantee threshold satisfied with probability \(\frac{1}{4} + \frac{3}{4} \cdot 0.4 = 0.55\) (which is more than \(0.5\)), but the expectation threshold is also \(\frac{1}{4} \cdot 1 + \frac{3}{4} \cdot \left(0.6 \cdot (-1) + 0.4 \cdot 9 \right) = 2.5\) (which is more than \(2\)).
    In fact, we can show that if the randomised strategy is to stay in \(\{v_0, v_1, v_2, v_3\}\) with probability \(q\) and to eventually go to \(v_4\) with probability \(1 - q\), then this strategy satisfies \(\BP\) for \(\frac{1}{6} \le q \le \frac{1}{3}\).
    This shows that deterministic strategies are not sufficient for the \(\BP\) synthesis problem and that randomised strategies are strictly more powerful.
    Moreover, in order to always stay in the MEC \(\{v_0, v_1, v_2, v_3\}\) with probability \(\frac{1}{4}\) (which is strictly between \(0\) and \(1\)), the strategy needs memory. 
    
    \begin{figure}[t]
        \centering
        \begin{tikzpicture}
            \node[state] (v0) {\(v_{0}\)};
            \node[random, draw, above right of=v0] (v1) {\(v_{1}\)};
            \node[random, draw, below right of=v0] (v2) {\(v_{2}\)};
            \node[state, below right of=v1] (v3) {\(v_{3}\)};
            \node[random, draw, right of=v3] (v4) {\(v_{4}\)};
            \node[state, above right of=v4] (v5) {\(v_{5}\)};
            \node[state, below right of=v4] (v6) {\(v_{6}\)};
            \node[random, draw, right of=v5] (v7) {\(v_{7}\)};
            \node[random, draw, right of=v6] (v8) {\(v_{8}\)};
            \draw 
                  (v0) edge[bend right=15] node[below right, pos=0.3, yshift=+1mm, xshift=-1mm]{\(+1\)} (v1)
                  (v1) edge[bend right=15] node[above left, pos=0.3]{\(+1, \color{blue}{.3}\)} (v0)
                  
                  (v0) edge[bend left=15] node[above right, pos=0.3, yshift=-1mm, xshift=-1mm]{\(+1\)} (v2)
                  (v2) edge[bend left=15] node[below left, pos=0.3]{\(+1, \color{blue}{.1}\)} (v0)
                  
                  (v1) edge[bend left=15] node[right, pos=0.3]{\(+1, \color{blue}{.7}\)} (v3)
                  (v3) edge[bend left=15] node[below left, pos=0.3, yshift=+1mm, xshift=+1mm]{\(+1\)} (v1)
                  
                  (v2) edge[bend right=15] node[below right, pos=0.3]{\(+1, \color{blue}{0.9}\)} (v3)
                  (v3) edge[bend right=15] node[above left, pos=0.3, yshift=-1mm, xshift=+1mm]{\(+1\)} (v2)
                  
                  (v3) edge node[above, pos=0.3]{\(+3\)} (v4)
                 
                  (v4) edge node[above left, pos=0.5]{\(0, \color{blue}{0.6}\)} (v5)
                  (v4) edge node[below left, pos=0.5]{\(0, \color{blue}{0.4}\)} (v6)

                  (v5) edge[bend left=15] node[above, pos=0.3]{\(-2\)} (v7)
                  (v7) edge[bend left=15] node[below, pos=0.3]{\(0, \color{blue}{1}\)} (v5)

                  (v6) edge[bend left=15] node[above, pos=0.3]{\(+10\)} (v8)
                  (v8) edge[bend left=15] node[below, pos=0.3]{\(+8, \color{blue}{1}\)} (v6)
                  
                  (v6) edge node[right, pos=0.4]{\(+20, \color{blue}{1}\)} (v7)
            ;
        \end{tikzpicture}
        \caption{An example of an MDP for \(\BPwArgs{0.5}{0}{2}\) with \(\WindowLength = 2\).}
        \label{fig:bp-example}
    \end{figure}
\end{example}

\begin{algorithm}
\caption{Window mean-payoff objective with probabilistic guarantee}\label{alg:wmp-probabilistic-guarantee}
    \begin{algorithmic}[1]
        \Require MDP \(\InputMDP\), vertex \(\VertexInitial \in V\), window length \(\WindowLength\),  probabilistic-guarantee threshold \((p, 0)\), and expectation threshold \(\ExpectationThreshold\)
        \Ensure {\Yes} if and only \(v \models \BP((p,0), \ExpectationThreshold)\)
        \State Compute MEC decomposition \(\MECs\) of \(\InputMDP\).
        \For{\(\MEC \in \MECs\) in \(\InputMDP\)}
            \State Compute the maximum \(\AlmostSureValue{M}\) such that the player almost-surely satisfies the threshold objective \(\ThresholdObjectiveFWMPL{\ge \AlmostSureValue{\MEC}}\) from every vertex \(v\) in the MDP restricted to the MEC \(\MEC\). \label{alg-line:compute_muM}
            \label{alg-line:probabilistic-sure-value-computation}
        \EndFor
        \State Construct \(\PrunedMDP\) from \(\InputMDP\) as follows: 
        \newline Collapse each MEC \(\MEC\) into a player vertex \(v_{\MEC}\), and add to it a self-loop with payoff \(\AlmostSureValue{\MEC}\).
        \State \Return \Yes\ if and only if \(\VertexInitial \models \BP((p, 0), \ExpectationThreshold)\) in \(\PrunedMDP\) for the \(\ObjectiveMP\) objective.%
        \label{alg-line:probabilistic-bp-mp-check}
    \end{algorithmic}
\end{algorithm}

\paragraph*{Description of \Cref{alg:wmp-probabilistic-guarantee}.}
We begin by finding the MEC decomposition \(\MECs\) of \(\InputMDP\).
Then, for every MEC \(\MEC \in \MECs\), we compute the maximum \(\ObjectiveFWMPL\)-value \(\AlmostSureValue{\MEC}\) that can be achieved almost-surely from a vertex in the MEC \(\MEC\).
This is well-defined as all vertices in a MEC have the same value since \(\ObjectiveFWMPL\) is a prefix-independent objective and every vertex in a MEC is almost-surely reachable from every other vertex in the MEC.
We then construct a new MDP \(\PrunedMDP\) by collapsing each MEC \(\MEC\) in \(\InputMDP\) to a vertex \(v_{\MEC}\).
That is, the set of out-edges of \(v_{\MEC}\) in \(\PrunedMDP\) is the union of the sets of out-edges of all the vertices \(v\) that go out from the MEC \(\MEC\) in \(\InputMDP\), and similarly, the set of in-edges of \(v_{\MEC}\) in \(\PrunedMDP\) is the union of the sets of in-edges of all the vertices in the MEC \(\MEC\) from vertices not in \(\MEC\).
We note that in the MDP \(\PrunedMDP\) with the collapsed MECs, each MEC has exactly one player vertex and a loop on it which contains a probabilistic vertex.
Finally, for each collapsed MEC \(\MEC\) in \(\PrunedMDP\), we add a self-loop with payoff \(\AlmostSureValue{\MEC}\).
In this collapsed MDP \(\PrunedMDP\), we solve the \(\BP\) problem for the \(\ObjectiveMP\) objective,
that is, classical mean-payoff objective,
using linear programming as in~\cite{CKK17} and briefly describe the details below.
If the vertex \(\VertexInitial\) does not belong to any MEC in \(\InputMDP\), then it appears in the collapsed MDP \(\PrunedMDP\) as well. 
Otherwise, if \(\VertexInitial\) belongs to some MEC \(\MEC_{\textsf{init}}\) in \(\InputMDP\), then for ease of notation, we continue to use \(\VertexInitial\) to represent the player vertex in \(\PrunedMDP\) obtained after collapsing \(\MEC_{\textsf{init}}\).

\begin{lemma}
\label{lem:probabilistic-reduction-to-mean-payoff}
    We have that \(\VertexInitial \models \BP((p, 0), \ExpectationThreshold)\) in \(\InputMDP\) for the \(\ObjectiveFWMPL\) objective if and only if \(\VertexInitial \models \BP((p, 0), \ExpectationThreshold)\) in \(\PrunedMDP\) for the \(\ObjectiveMP\) objective.
\end{lemma}
\begin{proof}
    We want to show that \(\VertexInitial \models \BP((p, 0), \ExpectationThreshold)\) in \(\InputMDP\) for the \(\ObjectiveFWMPL\) objective if and only if \(\VertexInitial \models \BP((p, 0), \ExpectationThreshold)\) in \(\PrunedMDP\) for the \(\ObjectiveMP\) objective.
    
    First, we show that 
    if \(\VertexInitial \models \BP((p, 0), \ExpectationThreshold)\) in \(\PrunedMDP\) for the \(\ObjectiveMP\) objective, then \(\VertexInitial \models \BP((p, 0), \ExpectationThreshold)\) in \(\InputMDP\) for the \(\ObjectiveFWMPL\) objective.
    An optimal strategy \(\StrategyMP\) for \(\BP((p, 0), \ExpectationThreshold)\) for \(\ObjectiveMP\) is one that tries to reach MECs with different probabilities, and then at some point, starts looping on those MECs. 
    Using \(\StrategyMP\), we construct a strategy \(\StrategyFWMPL\) that is optimal for \(\BP((p, 0), \ExpectationThreshold)\) for \(\ObjectiveFWMPL\) in \(\InputMDP\).
    The strategy \(\StrategyFWMPL\) mimics the strategy \(\StrategyMP\) in \(\InputMDP\) until \(\StrategyMP\) switches to looping. 
    When \(\StrategyMP\) reaches an MEC and switches to looping there, the strategy \(\StrategyFWMPL\) switches to \(\StrategyAlmostSure_{\FWMPL}\), an almost-sure winning strategy for \(\ThresholdObjectiveFWMPL{\ge \AlmostSureValue{\MEC}}\). 
    If the switch happens at a vertex in an MEC \(\MEC\) in \(\InputMDP\), then subsequently, the value of the run is \(\AlmostSureValue{\MEC}\) almost-surely.
    Thus, if threshold \(0\) is achieved with probability \(p\) in \(\PrunedMDP\), then the same threshold is achieved with the same probability in original MDP \(\InputMDP\), and if an expectation value \(\ExpectationThreshold\) is attained in \(\PrunedMDP\), then the same expectation value \(\ExpectationThreshold\) is also achieved in \(\InputMDP\).

    Now, we show the converse, that is 
    if \(\VertexInitial \models \BP((p, 0), \ExpectationThreshold)\) in \(\InputMDP\) for the \(\ObjectiveFWMPL\) objective, then \(\VertexInitial \models \BP((p, 0), \ExpectationThreshold)\) in \(\PrunedMDP\) for the \(\ObjectiveMP\) objective.
    Given an optimal strategy \(\StrategyFWMPL\) for the \(\BP((p, 0), \ExpectationThreshold)\) of \(\ObjectiveFWMPL\) in \(\InputMDP\), we construct an optimal strategy \(\StrategyMP\) for \(\BP((p, 0), \ExpectationThreshold)\) of \(\ObjectiveMP\) in \(\PrunedMDP\).
    By \Cref{lem:longrunMEC}, we have that playing according to the strategy \(\StrategyFWMPL\), the token eventually moves into an MEC \(\MEC\) from which it never exits.
    The strategy \(\StrategyMP\) mimics \(\StrategyFWMPL\) upto this point, after which \(\StrategyMP\) switches to looping on \(v_{\MEC}\).
    One can see that if the expectation and probability threshold are satisfied in \(\InputMDP\) for \(\ObjectiveFWMPL\), then they are also satisfied in \(\PrunedMDP\) for \(\ObjectiveMP\).
\end{proof}

\paragraph*{Solving the \(\BP\) problem for the \(\ObjectiveMP\) objective.}
We give a brief description of a linear program that can be used to determine if for a vertex \(v\) in \(\PrunedMDP\), if \(v \models \BP((p, 0), \ExpectationThreshold)\) for the \(\ObjectiveMP\) objective.
The linear program is simpler here than in~\cite{CKK17} since each MEC in the collapsed MDP has only one player vertex and one probabilistic vertex.

Recall that an optimal strategy for the \(\BP\) satisfaction of the \(\ObjectiveMP\) objective in the collapsed MDP \(\PrunedMDP\) is as follows: 
The strategy has two modes: first, it reaches MECs in \(\PrunedMDP\) with the appropriate probabilities such that both probability and expectation thresholds are satisfied.
If the token has reached a MEC \(\MEC\) with appropriate probability, then, the strategy switches to looping, that is, the token loops in \(\MEC\) forever. 

We solve the following linear program \(L\).

We have variables \(\LoopingProbability{v}{\LPYes}\), \(\LoopingProbability{v}{\LPNo}\) for all player vertices \(v \in \VerticesPlayer\) and \(\TransientTime{u}\) for all probabilistic vertices \(u \in \VerticesRandom\), all of which we require to be non-negative.
The variables have the following interpretations:
\begin{itemize}
    \item The variable \(\TransientTime{u}\) denotes the expected number of times that the probabilistic vertex \(u\) is visited before the strategy switches to looping.
    \item  The variables \(\LoopingProbability{v}{\LPYes}\) and \(\LoopingProbability{v}{\LPNo}\) together denote the probability that the strategy switches to looping from vertex \(v\).
    In particular, if \(\AlmostSureValue{v}\) is non-negative, then \(\LoopingProbability{v}{\LPYes}\) denotes the probability that the strategy switches to looping at \(v\) and \(\LoopingProbability{v}{\LPNo}\) is equal to zero.
    Otherwise, if \(\AlmostSureValue{v}\) is negative, then \(\LoopingProbability{v}{\LPYes}\) is equal to zero and 
    \(\LoopingProbability{v}{\LPNo}\) denotes the probability that the strategy switches to looping at \(v\).
    Note that \(\LoopingProbability{v}{\LPYes} + \LoopingProbability{v}{\LPNo}\) always gives the probability of switching to looping at \(v\).
\end{itemize}
The following are the equations in \(L\).
\begin{enumerate}
    \item Flow equations to conserve the probabilities of looping at each vertex \(v\).
    \[
        \text{For all } v \in \VerticesPlayer,\
        \mathbf{1}_{\VertexInitial}(v) + \sum_{u \in \VerticesRandom} \TransientTime{u} \cdot \ProbabilityFunction(u)(v) = \LoopingProbability{v}{\LPYes} + \LoopingProbability{v}{\LPNo} + \sum_{u \in \OutNeighbours{v} \intersection \VerticesRandom} \TransientTime{u}
    \]
    \item This constraint denotes that the strategy almost-surely switches to looping mode.
    \[
        \sum_{v \in \VerticesPlayer} (\LoopingProbability{v}{\LPYes} + \LoopingProbability{v}{\LPNo}) = 1
    \]
    \item 
    For each vertex \(v \in \VerticesPlayer\), the quantity \(\LoopingProbability{v}{\LPYes}+ \LoopingProbability{v}{\LPNo}\) denotes the probability of switching to looping at \(v\), and \(\PayoffFunction(v, v)\) denotes the payoff received on looping at \(v\).
    Thus, the following sum is the expected \(\ObjectiveMP\)-value, and the inequality ensures that the expectation threshold is satisfied.
    \[
        \sum_{v \in \VerticesPlayer} (\LoopingProbability{v}{\LPYes} + \LoopingProbability{v}{\LPNo}) \cdot \PayoffFunction(v, v) \ge \ExpectationThreshold
    \]
    \item This equation ensures that \(\LoopingProbability{v}{\LPYes}\) and \(\LoopingProbability{v}{\LPNo}\) satisfy the description given above, that is, if \(\AlmostSureValue{v}\) of a vertex \(v\) is non-negative, then the probability of switching in \(v\) is given by \(\LoopingProbability{v}{\LPYes}\).
    \[
    \text{For all } v \in \VerticesPlayer,\
        \LoopingProbability{v}{\LPYes} \cdot \PayoffFunction(v, v) \ge 0
    \]
    \[
    \text{For all } v \in \VerticesPlayer,\
        \LoopingProbability{v}{\LPNo} \cdot \PayoffFunction(v, v) \le 0
    \]
    \item This constraint ensures that the probability of looping at a vertex \(v\) with non-negative \(\AlmostSureValue{v}\) is at least \(p\).
    \[
        \sum_{v \in \VerticesPlayer} \LoopingProbability{v}{\LPYes} \ge p
    \]
\end{enumerate}

\paragraph*{Running time analysis.}
The MEC decomposition of \(\InputMDP\) can be done in polynomial time~\cite{CH14}.
For each MEC \(\MEC\) in \(\InputMDP\), computing the value of \(\AlmostSureValue{\MEC}\) can also be done in polynomial time using binary search~\cite{BGR19}.
Finally, to check if \(v \models \BP((p, 0),  \ExpectationThreshold)\) in \(\PrunedMDP\) for the \(\ObjectiveMP\), the algorithm solves the linear program \(L\), which can be done in polynomial time since the number of variables and constraints in \(L\) is polynomial in the size of the input~\cite{Khachiyan1980}.

\paragraph*{Memory requirements of optimal strategies.}
\Cref{fig:bp-example} shows that, in general, optimal strategies of the player requires memory and randomization, and thus, deterministic strategies do not suffice. 
However, finite memory suffices.
In particular, before reaching the MECs in \(\InputMDP\), the strategy \(\StrategyFWMPL\) may need memory as well as randomisation as illustrated in \Cref{ex:bp-example}, and the strategy \(\StrategyAlmostSure_{\FWMPL}\) can be a deterministic strategy with memory size \(\WindowLength\)~\cite{DGG24}.

\begin{theorem}%
\label{thm:fwmp-probability-result}
    The \(\BP\) synthesis problem for the \(\ObjectiveFWMPL\) objective is in \(\PTime\), and if \(v \models \BP((p, \GuaranteeThreshold), \ExpectationThreshold)\), then there exists an optimal finite-memory randomised strategy from \(v\).
\end{theorem}

\subsection{Almost-sure guarantee}
In this section, we solve the \(\BAS\) synthesis problem, that is, we decide, given an MDP \(\InputMDP\), a vertex \(\VertexInitial\), and an expectation threshold \(\ExpectationThreshold\), if \(\VertexInitial\) satisfies \(\BAS(0, \ExpectationThreshold)\) in \(\InputMDP\).
In the \(\BAS\) synthesis problem, similar to the \(\BWC\) synthesis problem, we are interested in the case when \(\ExpectationThreshold > 0\), as otherwise satisfying \(\ThresholdObjectiveFWMPL{\ge 0}\) almost-surely implies that the expected \(\ObjectiveFWMPL\)-value is also at least \(0\) and hence the expectation threshold is trivially satisfied.

The \(\BASwArgs{0}{\ExpectationThreshold}\) synthesis problem is a special case of the \(\BP((p, 0), \ExpectationThreshold)\) problem, when \(p = 1\).
We get better bounds for the memory size of the optimal strategies for \(\BAS\) as compared to \(\BP\) in general, and moreover, deterministic strategies suffice for \(\BAS\).
Note that \Cref{alg:wmp-probabilistic-guarantee} works for \(\BAS(0, \ExpectationThreshold)\) as well if we let \(p = 1\).
In Line~\ref{alg-line:probabilistic-bp-mp-check} in \Cref{alg:wmp-probabilistic-guarantee}, we need to check if \(\VertexInitial \models \BP((1, 0), \ExpectationThreshold)\) in the collapsed MDP \(\PrunedMDP\) for the \(\ObjectiveMP\) objective.
That is, we need to check if \(\VertexInitial \models \BAS(0, \ExpectationThreshold)\) in \(\PrunedMDP\) for \(\ObjectiveMP\).
Instead of using the linear program defined in \Cref{sec:bp-fwmp} to do the check,  we use a reduction of the \(\BAS\) problem for \(\ObjectiveMP\) to the problem of standard expected \(\ObjectiveMP\)-value that is described in~\cite{CKK17}.
The reduction is as follows: 
We prune all vertices from \(\PrunedMDP\) from which the player cannot almost-surely achieve non-negative \(\ObjectiveMP\)-value to get an MDP \(\PrunedMDPWithSelfLoops\).
By analysing the MECs, we can check in \(\PTime\) if  the player can almost-surely achieve non-negative \(\ObjectiveMP\)-value from a vertex.
If \(\VertexInitial\) is pruned away, then it does not satisfy  \(\BAS(0, \ExpectationThreshold)\) for \(\ObjectiveMP\) in \(\PrunedMDP\).
Otherwise, \(\VertexInitial\) satisfies \(\BAS(0, \ExpectationThreshold)\) in \(\PrunedMDP\) if and only if in the pruned MDP \(\PrunedMDPWithSelfLoops\), the expected \(\ObjectiveMP\)-value of \(\VertexInitial\) in \(\PrunedMDPWithSelfLoops\) is at least \(\ExpectationThreshold\).
The correctness of the algorithm follows from Lemma~\ref{lem:probabilistic-reduction-to-mean-payoff} by setting \(p=1\).

\paragraph*{Memory requirements of optimal strategies.}
Let \(\StrategyMP\) be a memoryless deterministic optimal strategy for expected classical mean-payoff objective \(\ObjectiveMP\) in \(\PrunedMDPWithSelfLoops\).
An optimal strategy \(\Strategy^{*}\) for \(\BAS(0, \ExpectationThreshold)\) for \(\ObjectiveFWMPL\) in \(\InputMDP\) can be constructed by first mimicking \(\StrategyMP\) until \(\StrategyMP\) switches to looping. 
If the token is on a vertex \(v\) when this switch happens in \(\PrunedMDPWithSelfLoops\), then \(\Strategy^{*}\) should switch to mimicking an optimal strategy for the almost-sure satisfaction of the threshold objective \(\ThresholdObjectiveFWMPL{\ge \AlmostSureValue{M}}\) in \(\InputMDP\).

Since there exist deterministic memoryless optimal strategies  for expectation of \(\ObjectiveMP\)~\cite{Puterman94}, and there exist deterministic optimal strategies with memory size at most \(\WindowLength\) for the almost-sure satisfaction of \(\ThresholdObjectiveFWMPL{\ge \AlmostSureValue{M}}\)~\cite{DGG24}, we get that there exist deterministic optimal strategies with memory size at most \(\WindowLength\) for \(\BAS(0, \ExpectationThreshold)\).

\begin{theorem}%
\label{thm:fwmp-almost-sure-result}
    The \(\BAS\) synthesis problem for the \(\ObjectiveFWMPL\) objective is in \(\PTime\), and if \(v \models \BAS(\GuaranteeThreshold, \ExpectationThreshold)\), then there exists an optimal deterministic strategy of memory size at most \(\WindowLength\) from \(v\).
\end{theorem}

\section{Expected bounded window mean-payoff value with guarantees}
\label{sec:bwmp}
In this section, we study the expectation maximisation problem with sure, almost-sure, and probabilistic guarantees for the bounded window mean-payoff objective.
The algorithms are similar to those for the fixed window mean-payoff objective described in the previous section.
We only highlight the main differences here.
Note that all our algorithms require solving two-player games with either \(\{\ObjectiveBWMP \ge 0\}\) or \(\{\ObjectiveMP \ge 0\}\) objective.
While two-player games with \(\ThresholdObjectiveFWMPL{\ge 0}\) objective can be solved in \(\PTime\), solving two-player games with \(\{\ObjectiveBWMP \ge 0\}\) objective or \(\{\ObjectiveMP \ge 0 \}\) objective is in \(\NP \cap \coNP\)~\cite{BGR19}\footnote{In a two-player game, the winning region for the \(\{\ObjectiveBWMP \ge 0\}\)~\cite{BGR19} objective may strictly contain the winning region for the \(\BWMP(0)\) objective~\cite{CDRR15} as shown in Proposition~\ref{prop:BWMP}.}.

\paragraph*{Maximising expectation with sure guarantee}
The algorithm here is similar to \Cref{alg:wmp-sure-guarantee}.
As stated above, computing \(\SureWinningSet{\BWMP}{0}\) is in \(\NP \cap \coNP\).
The MDP \(\PrunedMDP \define \InputMDP \restriction \SureWinningSet{\BWMP}{0}\) is the MDP obtained by restricting \(\InputMDP \) to \(\SureWinningSet{\BWMP}{0}\).
For every \(\Vertex \in \VerticesPlayer\), we compute the maximum \(\SureValue{v}\) such that \(v\) belongs to the sure winning region in \(\PrunedMDP\) for the \(\{\ObjectiveBWMP \ge \SureValue{v}\}\) objective.
Recall that for a run \(\Run\), we have \(\ObjectiveBWMP(\Run) = \sup \{\Threshold \in \Reals \suchthat \exists \WindowLength \ge 0 : \Run \in \FWMP(\WindowLength, \Threshold) \} \).
The maximum \(\SureValue{v}\) can be computed by solving the two-player game \(\Game_\MDP\) with the classical mean-payoff objective from \(v\)~\cite{BGR19}\footnote{Two-player games with the \(\BWMP(\Threshold)\) objective are solved by reducing it to two-player games with total payoff~\cite{CDRR15}.}.
This value thus equals the mean payoff of a cycle in \(\Game_\MDP\) and is of the form \(\frac{a}{b}\) where \(a \in \{0, \ldots, W \cdot \abs{\Vertices}\}\) and \(b \in \{1, \ldots \abs{\Vertices}\}\).
Thus, a binary search is done over \(W \cdot |\Vertices|^2\) many values and hence the two-player game is solved polynomially many times.

The strategy \(\Strategy^{*}_{\epsilon}\) for \(\BWC\) synthesis for the \(\BWMP\) objective is similar to \(\FWMPL\) objective with the difference that the strategy for achieving \(\{\ObjectiveBWMP \ge \SureValue{v}\}\) from \(v\) is memoryless~\cite{BGR19}.

\paragraph*{Maximising expectation with probabilistic guarantee.}
The algorithm for \(\BPwArgs{p}{0}{\ExpectationThreshold}\) synthesis for the \(\ObjectiveBWMP\) objective is almost the same as Algorithm~\ref{alg:wmp-probabilistic-guarantee} with the difference that in Line~\ref{alg-line:compute_muM} to compute the maximum \(\AlmostSureValue{M}\), we use \(\{\ObjectiveBWMP \ge \AlmostSureValue{\MEC}\}\) instead of \(\ThresholdObjectiveFWMPL{\ge \AlmostSureValue{\MEC}}\).
The problem of determining if the player almost-surely satisfies the threshold objective \(\{\ObjectiveBWMP \ge \AlmostSureValue{\MEC}\}\) from every vertex \(\Vertex\) in the MDP \(\InputMDP\) restricted to the MEC \(\MEC\) is in \(\NP \cap \coNP\) by reducing it to a polynomial number of calls to two-player classical mean-payoff games~\cite{BGR19}.
Here also an optimal strategy for \(\BPwArgs{p}{0}{\ExpectationThreshold}\) synthesis for the \(\ObjectiveBWMP\) objective may need both memory and randomisation.

\paragraph*{Maximising expectation with almost-sure guarantee.}
Again, the algorithm is similar to the \(\BAS\) synthesis problem for the \(\ObjectiveFWMPL\) objective.
The \(\BAS\) synthesis problem for the \(\ObjectiveBWMP\) objective is in \(\NP \cap \coNP\) since computing \(\AlmostSureValue{\MEC}\) for each vertex is in \(\NP \cap \coNP\).
In contrast to almost-sure satisfaction of \(\{\ObjectiveFWMPL \ge \AlmostSureValue{\MEC}\}\), optimal deterministic memoryless strategies exist for almost-sure satisfaction of \(\{\ObjectiveBWMP \ge \AlmostSureValue{\MEC}\}\)~\cite{BGR19}.
It follows that optimal deterministic memoryless strategies exist for the \(\BAS\) synthesis problem for the \(\ObjectiveBWMP\) objective.

We thus have the following.
\begin{theorem}%
\label{thm:BWMP}
    The \(\BWC\) synthesis problem, the \(\BP\) synthesis problem, and the \(\BAS\) synthesis problem for the \(\ObjectiveBWMP\) objective are in \(\NP \cap \coNP\), and 
    \begin{enumerate} [(i)]
        \item if \(v \models \BWC(\GuaranteeThreshold, \ExpectationThreshold)\), then for every \(\epsilon > 0\), there exists an \(\epsilon\)-optimal finite-memory deterministic strategy from \(v\).
        \item if \(v \models \BP((p, \GuaranteeThreshold), \ExpectationThreshold)\), then there exists an optimal finite-memory randomised strategy from \(v\).
        \item if \(v \models \BAS(\GuaranteeThreshold, \ExpectationThreshold)\), then there exists an optimal deterministic memoryless strategy from \(v\).
    \end{enumerate}
\end{theorem}
Thus the complexities achieved are no more than sure satisfaction of the \(\{\ObjectiveBWMP \ge 0\}\) in a two-player game or expectation maximisation for the \(\ObjectiveBWMP\) objective in an MDP.

\section{Conclusion}%
\label{sec:conc}
Expectation maximisation with guarantees is a natural problem of importance and interest and appears in various real-world contexts.
Further, window mean-payoff objective strengthens classical mean-payoff objective and prevents some undesired behaviours of classical mean payoff from happening.
We have shown that the \(\BWC\), \(\BAS\), and the \(\BP\) synthesis of fixed window mean-payoff objectives for MDPs are in \(\PTime\) while the problems are in \(\NP \cap \coNP\) for the bounded window mean-payoff objective.
We note that the \(\BWC\) synthesis problem for classical mean payoff is already in \(\NP \cap \coNP\)~\cite{BFRR17,CR15}.
Our results establish that these problems can be solved at no additional cost than solving the expectation problem for the window mean-payoff objectives while not providing any guarantee, or solving the window mean-payoff objectives with the guarantees while disregarding any requirement on the expected behaviour.

As part of future work, we would like to extend the notion of beyond worst-case and beyond almost-sure to other finitary objectives.
It would also be interesting to study these problems in the context of stochastic games.

\bibliography{mybib}

\begin{thebibliography}{10}

\bibitem{AKV16}
S.~Almagor, O.~Kupferman, and Y.~Velner.
\newblock Minimizing expected cost under hard boolean constraints, with applications to quantitative synthesis.
\newblock In J.~Desharnais and R.~Jagadeesan, editors, {\em 27th International Conference on Concurrency Theory, {CONCUR} 2016, August 23-26, 2016, Qu{\'{e}}bec City, Canada}, volume~59 of {\em LIPIcs}, pages 9:1--9:15. Schloss Dagstuhl - Leibniz-Zentrum f{\"{u}}r Informatik, 2016.

\bibitem{AG11}
K.R. Apt and E.~Gr{\"a}del.
\newblock {\em Lectures in Game Theory for Computer Scientists}.
\newblock Cambridge University Press, 2011.

\bibitem{BK08}
C.~Baier and J{-}P. Katoen.
\newblock {\em Principles of model checking}.
\newblock {MIT} Press, 2008.

\bibitem{BGR20}
R.~Berthon, S.~Guha, and J{-}F. Raskin.
\newblock Mixing probabilistic and non-probabilistic objectives in markov decision processes.
\newblock In {\em {LICS} '20}, pages 195--208. {ACM}, 2020.

\bibitem{BKW24}
R.~Berthon, J{-}P. Katoen, and T.~Winkler.
\newblock Markov decision processes with sure parity and multiple reachability objectives.
\newblock In L.~Kov{\'{a}}cs and A.~Sokolova, editors, {\em Reachability Problems - 18th International Conference, {RP} 2024, Vienna, Austria, September 25-27, 2024, Proceedings}, volume 15050 of {\em Lecture Notes in Computer Science}, pages 203--220. Springer, 2024.

\bibitem{BRR17}
R.~Berthon, M.~Randour, and J{-}F Raskin.
\newblock Threshold constraints with guarantees for parity objectives in markov decision processes.
\newblock In Ioannis Chatzigiannakis, Piotr Indyk, Fabian Kuhn, and Anca Muscholl, editors, {\em 44th International Colloquium on Automata, Languages, and Programming, {ICALP} 2017, July 10-14, 2017, Warsaw, Poland}, volume~80 of {\em LIPIcs}, pages 121:1--121:15. Schloss Dagstuhl - Leibniz-Zentrum f{\"{u}}r Informatik, 2017.

\bibitem{Billingsley86}
P.~Billingsley.
\newblock {\em Probability and Measure}.
\newblock John Wiley and Sons, second edition, 1986.

\bibitem{BGR19}
B.~Bordais, S.~Guha, and J{-}F. Raskin.
\newblock Expected window mean-payoff.
\newblock In {\em {FSTTCS}}, volume 150 of {\em LIPIcs}, pages 32:1--32:15, 2019.

\bibitem{BKN16}
T.~Br{\'a}zdil, A.~Ku{\v{c}}era, and P.~Novotn{\`y}.
\newblock Optimizing the expected mean payoff in energy markov decision processes.
\newblock In {\em International Symposium on Automated Technology for Verification and Analysis}, pages 32--49. Springer, 2016.

\bibitem{BDOR20}
T.~Brihaye, F.~Delgrange, Y.~Oualhadj, and M.~Randour.
\newblock {Life is Random, Time is Not: Markov Decision Processes with Window Objectives}.
\newblock {\em {Logical Methods in Computer Science}}, {Volume 16, Issue 4}, 12 2020.

\bibitem{BFRR17}
V.~Bruyère, E.~Filiot, M.~Randour, and J-F. Raskin.
\newblock Meet your expectations with guarantees: Beyond worst-case synthesis in quantitative games.
\newblock {\em Information and Computation}, 254:259--295, 2017.

\bibitem{CDRR15}
K.~Chatterjee, L.~Doyen, M.~Randour, and J-F. Raskin.
\newblock Looking at mean-payoff and total-payoff through windows.
\newblock {\em Information and Computation}, 242:25--52, 2015.

\bibitem{CENR18}
K.~Chatterjee, A.~Elgy{\"{u}}tt, P.~Novotn{\'{y}}, and O.~Rouill{\'{e}}.
\newblock Expectation optimization with probabilistic guarantees in {POMDPs} with discounted-sum objectives.
\newblock In J{\'{e}}r{\^{o}}me Lang, editor, {\em Proceedings of the Twenty-Seventh International Joint Conference on Artificial Intelligence, {IJCAI} 2018, July 13-19, 2018, Stockholm, Sweden}, pages 4692--4699. ijcai.org, 2018.

\bibitem{CH14}
K.~Chatterjee and M.~Henzinger.
\newblock Efficient and dynamic algorithms for alternating b{\"{u}}chi games and maximal end-component decomposition.
\newblock {\em J. {ACM}}, 61(3):15:1--15:40, 2014.
\newblock \href {https://doi.org/10.1145/2597631} {\path{doi:10.1145/2597631}}.

\bibitem{CHH09}
K.~Chatterjee, T.~A. Henzinger, and F.~Horn.
\newblock Stochastic games with finitary objectives.
\newblock In {\em MFCS}, pages 34--54. Springer Berlin Heidelberg, 2009.

\bibitem{CKK17}
K.~Chatterjee, Z.~Křetínská, and J.~Křetínský.
\newblock {Unifying Two Views on Multiple Mean-Payoff Objectives in Markov Decision Processes}.
\newblock {\em {Logical Methods in Computer Science}}, {Volume 13, Issue 2}, July 2017.

\bibitem{CNPRZ17}
K.~Chatterjee, P.~Novotn{\'{y}}, G.~A. P{\'{e}}rez, J{-}F. Raskin, and D.~Zikelic.
\newblock Optimizing expectation with guarantees in {POMDPs}.
\newblock In Satinder Singh and Shaul Markovitch, editors, {\em Proceedings of the Thirty-First {AAAI} Conference on Artificial Intelligence, February 4-9, 2017, San Francisco, California, {USA}}, pages 3725--3732. {AAAI} Press, 2017.

\bibitem{CP19}
K.~Chatterjee and N.~Piterman.
\newblock Combinations of qualitative winning for stochastic parity games.
\newblock In {\em 30th International Conference on Concurrency Theory, {CONCUR} 2019, August 27-30, 2019, Amsterdam, the Netherlands}, volume 140 of {\em LIPIcs}, pages 6:1--6:17. Schloss Dagstuhl - Leibniz-Zentrum f{\"{u}}r Informatik, 2019.

\bibitem{CR15}
L.~Clemente and J.-F. Raskin.
\newblock Multidimensional beyond worst-case and almost-sure problems for mean-payoff objectives.
\newblock In {\em 2015 30th Annual ACM/IEEE Symposium on Logic in Computer Science}, pages 257--268. IEEE, 2015.

\bibitem{DGG24}
L.~Doyen, P.~Gaba, and S.~Guha.
\newblock Stochastic window mean-payoff games.
\newblock In {\em {FoSSaCS} Part I}, volume 14574 of {\em LNCS}, pages 34--54. Springer, 2024.

\bibitem{EM79}
A.~Ehrenfeucht and J.~Mycielski.
\newblock Positional strategies for mean payoff games.
\newblock {\em Int. Journal of Game Theory}, 8(2):109--113, 1979.

\bibitem{GGR18}
G.~Geeraerts, S.~Guha, and J{-}F. Raskin.
\newblock Safe and optimal scheduling for hard and soft tasks.
\newblock In {\em {FSTTCS}}, pages 36:1--36:22, 2018.

\bibitem{GTW02}
E.~Gr{\"{a}}del, W.~Thomas, and T.~Wilke, editors.
\newblock {\em Automata, Logics, and Infinite Games: {A} Guide to Current Research [outcome of a Dagstuhl seminar, February 2001]}, volume 2500 of {\em Lecture Notes in Computer Science}. Springer, 2002.

\bibitem{Khachiyan1980}
L.G. Khachiyan.
\newblock Polynomial algorithms in linear programming.
\newblock {\em USSR Computational Mathematics and Mathematical Physics}, 20(1):53--72, 1980.

\bibitem{Puterman94}
M.L. Puterman.
\newblock {\em Markov decision processes: Discrete stochastic dynamic programming}.
\newblock John Wiley and Sons, 1994.

\bibitem{Vardi85}
M.~Y. Vardi.
\newblock Automatic verification of probabilistic concurrent finite-state programs.
\newblock In {\em 26th Annual Symposium on Foundations of Computer Science, Portland, Oregon, USA, 21-23 October 1985}, pages 327--338. {IEEE} Computer Society, 1985.

\bibitem{ZP96}
U.~Zwick and M.~Paterson.
\newblock The complexity of mean payoff games on graphs.
\newblock {\em Theor. Comput. Sci.}, 158(1{\&}2):343--359, 1996.

\end{thebibliography}

\end{document}